\documentclass[a4paper,UKenglish,cleveref,pdfa,authorcolumns]{lipics-v2021}%anonymous
\hideLIPIcs

\usepackage{cases}
\usepackage{xfrac}
\usepackage{wrapfig}

\usepackage{algorithm}
\usepackage{algorithmic}

\newcommand{\LINECOMMENT}[1]{\STATE \textit{//~#1}}

\newcommand{\R}{\ensuremath{\mathbb{R}}}
\DeclareMathOperator*{\argmin}{arg\,min}
\DeclareMathOperator*{\odil}{odil}
\DeclareMathOperator{\imax}{imax}
\newcommand{\bigo}{\ensuremath{\mathcal{O}}}

\usepackage{tikz}
\usetikzlibrary{calc}
\usetikzlibrary{positioning}
\usepackage{tkz-euclide}

\tikzset{
	graph node/.style = {
		circle, 
		thick,
		rounded corners, 
		fill=lightgray, 
		draw=black,
		align=center,
		font=\normalsize,
	},
	dot/.style={
		circle, 
		fill, 
		inner sep=2pt,
	}
}

\tikzset{
	graph/plain/.style={
		draw=none, 
		fill=none, 
		text=black,
	},
}

\tikzset{
	graph edge/.style={
		draw= black,
		thick,
		every node/.style={graph/plain}%labels
	},
}

\tikzset{
	graph/.style={
		node distance=3em,
		every node/.style={graph node, minimum size=1.5em},
		every edge/.style={graph edge},
	},
    directed/.style={%
		every edge/.append style={-Stealth},
	},%
}

\newcommand{\replacecolor}[2]
{%
	\expandafter\let\csname\string\color @#1\expandafter\endcsname
	\csname\string\color @#2\endcsname%
}

\replacecolor{firstColor}{cyan}
\replacecolor{secondColor}{red}
\replacecolor{thirdColor}{green}

\title{Oriented Spanners}

\author{Kevin Buchin}{Technical University of Dortmund, Germany}{kevin.buchin@tu-dortmund.de}{https://orcid.org/0000-0002-3022-7877}{}
\author{Carolin Rehs}{Technical University of Dortmund, Germany}{carolin.rehs@tu-dortmund.de}{https://orcid.org/0000-0002-8788-1028}{}
\author{Joachim Gudmundsson}{University of Sydney, Australia}{joachim.gudmundsson@sydney.edu.au}{https://orcid.org/0000-0002-6778-7990}{}
\author{André {van Renssen}}{University of Sydney, Australia}{andre.vanrenssen@sydney.edu.au}{https://orcid.org/0000-0002-9294-9947}{}
\author{Antonia Kalb}{Technical University of Dortmund, Germany}{antonia.kalb@tu-dortmund.de}{https://orcid.org/0009-0009-0895-8153}{}
\author{Sampson Wong}{BARC, University of Copenhagen, Denmark}{sawo@di.ku.dk}{https://orcid.org/0000-0003-3803-3804}{}
\author{Aleksandr Popov}{Technical University of Eindhoven,\newline The Netherlands}{a.popov@tue.nl}{https://orcid.org/0000-0002-0158-1746}{Supported by the Dutch Research Council (NWO) under the project number 612.001.801.}

\authorrunning{K. Buchin, J. Gudmundsson, A. Kalb, A. Popov, C. Rehs, A. van Renssen, S. Wong}
\Copyright{Kevin Buchin, Joachim Gudmundsson, Antonia Kalb, Aleksandr Popov, Carolin Rehs, André van Renssen, and Sampson Wong}

\ccsdesc[500]{Theory of computation~Computational geometry}
\keywords{computational geometry, spanner, oriented graph, greedy triangulation}

\relatedversion{A preliminary version of this paper appeared in 31st Annual European Symposium on Algorithms (ESA 2023)~\cite{BuchinGKPRRW23}} %optional, e.g. full version hosted on arXiv, HAL, or other respository/website
\nolinenumbers %uncomment to disable line numbering

\begin{document}
	\thispagestyle{empty}
	\maketitle
	
	\begin{abstract}
		Given a point set $P$ in the Euclidean plane and a parameter $t$, we define an \emph{oriented $t$-spanner} $G$ as an oriented subgraph of the complete bi-directed graph such that for every pair of points, the shortest closed walk in $G$ through those points is at most a factor $t$ longer than the shortest cycle in the complete graph on $P$. We investigate the problem of computing sparse graphs with small oriented dilation.
		
		As we can show that minimising oriented dilation for a given number of edges is NP-hard in the plane, we first consider one-dimensional point sets. While obtaining a $1$-spanner in this setting is straightforward, already for five points such a spanner has no plane embedding with the leftmost and rightmost point on the outer face. 
		This leads to restricting to oriented graphs with a one-page book embedding on the one-dimensional point set. For this case we present a dynamic program to compute the graph of minimum oriented dilation that runs in $\bigo(n^7)$ time for $n$ points, and a greedy algorithm that computes a $5$-spanner in $\bigo(n\log n)$ time.
		
		Expanding these results finally gives us a result for two-dimensional point sets: we prove that for convex point sets the greedy triangulation results in a plane oriented $t$-spanner  with $t=7.2 \cdot t_g$, where $t_g$ is an upper bound on the dilation of the greedy triangulation.
	\end{abstract}
	
	%%%%%%%%%%%%%%%%%%%%%%%%%%%%%%%
	% \newpage
	% \setcounter{page}{1}
	\section{Introduction}
	
	% Define spanners in general
	Computing geometric spanners is an extensively studied problem~\cite{Bose.2013,Narasimhan.2007}. Directed geometric spanners have also been considered~\cite{AkitayaBB22}.
	%,Peleg.1989}: kevin: we switched to the geometric setting without warning. I think it is is better to do it directly
Given a point set $P \subset \R^d$ and a parameter $t$, a \emph{directed $t$-spanner} $G=(P,E)$ is a subgraph of the complete bi-directed geometric graph on $P$ such that for every pair of points $p,p'$, the shortest path in $G$ is at most a factor $t$ longer than the shortest path in the complete graph, that is, $|p-p'|$. The \emph{dilation} of $G$ then is the smallest such $t$. Formally, $t=\max\bigl\{\frac{|d_G(p,p')|}{|p-p'|}\bigm| p,p'\in P\bigr\}$, where $d_G(p,p')$ denotes the \emph{shortest path from $p$ to $p'$ in $G$}.

(Directed) geometric spanners have a wide range of applications, ranging from wireless ad-hoc networks~\cite{burkhart2004does,schindelhauer2007geometric} to robot motion planning~\cite{dobson2014sparse} and the analysis of road networks~\cite{aronov2011connect,Eppstein00}. In all of these applications one might want to avoid adding the edge $(v,u)$ if the edge $(u,v)$ was included: in ad-hoc networks this may reduce interference, in motion planning it may reduce congestion and simplify collision avoidance, in road networks this corresponds to one-way roads or tracks, which may be necessary because of space limitations, and in communication networks one could require  two neighbouring devices not to exchange data by the same (bi-directional) direct connection, for example, in two-way authentication.

This motivates our study of \emph{oriented graphs} as spanners, i.e.\ directed spanners $G=(P,E)$ where  $(p,p')\in E$ implies  $(p',p)\notin E $. 
With this restriction, if the %relation of $d_G(p,p')$ to $|p - p'|$ is minimized by adding the 
edge $(p,p')$ is added, the dilation in the other direction is never~$1$. Even worse, 
given a set $P$ of three points, where $p$ and $p'$ are very close to each other and $p''$ is far away from both, any oriented graph will have arbitrarily high dilation for either $(p,p')$ or $(p',p)$ (see \Cref{fig:motivation-odil}). 
Therefore, considering the dilation for an oriented graph as $t=\max\bigl\{\frac{|d_G(p,p')|}{|p-p'|}\bigm| p,p'\in P\bigr\}$ would not tell us much about the quality of the spanner.
To obtain meaningful results, we define \emph{oriented dilation}.

\begin{figure}[ht]
	\centering
	\begin{tikzpicture}[directed, scale=.4]
		\node[dot,label={-180:{$p\phantom{'}$}}] (v1) at (0,0) {};
		\node[dot, label={-180:{$p'$}}] (v2) at (0,2) {};
		\node[dot, label={0:{$p''$}}] (v3) at (15,1) {};
		\draw  (v1) edge (v2);
		\draw  (v2) edge (v3);
		\draw  (v3) edge (v1);
	\end{tikzpicture}
	\caption{If $p$ and $p'$ are very close to each other and $p''$ is far away from both, any oriented graph will have arbitrarily high (directed) dilation.}
	\label{fig:motivation-odil}
\end{figure}

%Antonia Where to place these definitions?
A \emph{walk} is defined as  a sequence of points and edges of a graph. A walk is called \emph{closed} if it starts and ends at the same point. The \emph{length of a walk} is the sum of the lengths of its edges. A \emph{cycle} is a walk where all edges and points are distinct except the start and end point which are the same. A \emph{path} is a walk where all points and edges are distinct. 

By $C_G(p,p')$ we denote the \emph{shortest closed walk}\footnote{
	In the preliminary version of this paper~\cite{BuchinGKPRRW23}, the term \emph{shortest oriented cycle} is used. Since $C_G(p,p')$ can use edges twice, we have changed the term to \emph{shortest closed walk} (in an oriented graph $G$).} 
containing the points $p$ and $p'$ in an oriented graph $G$. The \emph{optimal cycle $\Delta(p,p')$} for $p$ and $p'$ is the shortest cycle containing $p$ and $p'$ in the complete undirected graph on $P$. Notice, $\Delta(p,p')$ is the triangle $\Delta_{pp'p''}$  with $p''=\argmin_{p^*\in P \setminus \{p, p'\}} \bigl(|p-p^*|+|p^*-p'|\bigr)$.

\begin{definition}[oriented dilation]\label{def:odilation}
	Given a point set $P$ and an oriented graph $G$ on $P$, the \emph{oriented dilation of two points} $p,p'\in P$ is defined as 
	\[
	\odil(p,p')=\frac{|C_G(p,p')|}{|\Delta(p,p')|}.
	\]
	The dilation $t$ of an oriented graph is defined as % $G=(P,E)$ with $E\subset P\times P$     
	$t=\max\{\odil(p,p')\mid p,p'\in P\}.$ 
\end{definition}
An oriented graph with dilation at most $t$ is called an \emph{oriented $t$-spanner}. 

We frequently contrast our results to known results on undirected geometric spanners, and refer to the known results by using the adjective \emph{undirected}.
Our new measure for oriented graphs is similar to the definition of dilation in round trip spanners~\cite{Cowen.1999,Cowen.2004} that has been considered in the setting of (non-geometric) directed graph spanners. But round trip spanners require a starting graph, and using the complete bi-directed geometric graph as input would not give meaningful results. 

In this paper, we initiate the study of oriented spanners. As is common for spanners, our general goal is to obtain \emph{sparse} spanners, i.e.\ with linear number of edges. The goal can be achieved by bounding the number of edges explicitly or by restricting to a class of sparse graphs like plane graphs~\cite{Bose.2013}. We refer to a spanner as a \emph{minimum (oriented) spanner} if it minimises $t$ under the given restriction.

It is known that computing a minimum undirected spanner with at most $n-1$ edges, i.e.\ a minimum dilation tree, is NP-hard~\cite{Giannopoulos.2010}. The corresponding question for oriented spanners asks for the \emph{minimum dilation cycle}. We prove this problem to be NP-hard in \Cref{sec:2D-hardness}.

The problem of computing the minimum undirected spanner restricted to the class of plane straight-line graphs is called the \emph{minimum dilation triangulation} problem; its hardness is still open~\cite{Eppstein00,Giannopoulos.2010}, but it is conjectured to be NP-hard~\cite{brandt-mdt-2014}. As this undirected problem can be emulated in the oriented setting by suitable vertex gadgets (\Cref{observation:minimum_dilation_triangulation}), it is unlikely that finding a \emph{minimum plane straight-line (oriented) spanner} can be done efficiently.

Therefore, in \Cref{sec:1D}, we start with one-dimensional point sets.
For such a point set with $n$ points, we can give a $1$-spanner with $3n-6$ edges. However, if we are interested in a one-dimensional result analogous to minimum plane spanners, this spanner is not suitable: it has no plane embedding with leftmost and rightmost point on the outer face.
Therefore, we restrict our attention to a graph class that is closer to the plane case for two-dimensional point sets: one-page book embeddings. 

We show how to compute a $t$-spanner which is a one-page plane book embedding for a one-dimensional point set in \Cref{sec:1D-1PPB}. We prove that with a greedy algorithm, we can always generate such a $t$-spanner with $t \leq 5$ in $\bigo(n \log n) $ time. An optimal one-page plane spanner can be computed in $\bigo(n^7)$ time.

As a one-page plane spanner is outerplanar, this particular class of graphs is also motivated by the problem of finding a  minimum plane spanner for points in convex position. Using these results, in \Cref{sec:2D-greedy}, we show that suitably orienting the greedy triangulation leads to oriented $\bigo(1)$-spanners for two-dimensional point sets in convex position. For general (non-convex) point sets, there are examples where all orientations of the greedy triangulation have large dilation. 

The greedy triangulation fulfils the $\alpha$-diamond-property~\cite{DasJ89}, and all triangulations with this property are undirected $\bigo(1)$-spanners. This raises the question whether all triangulations fulfilling this property are also oriented $\bigo(1)$-spanners for convex point sets. In \Cref{sec:other_triangulations} we answer this question negatively. 

\begin{table}[ht]
	\centering
	\caption{Overview of the results of the paper}
	\label{fig:overview}
	\begin{tabular}{c|c|c|c|c}
		\textbf{Point set } & \textbf{Spanner type} & \textbf{Dilation} & \textbf{Time complexity} & \textbf{Reference}   \\
		%sw Point set
		\hline
		$2$-dim. & sparse & minimum & NP-hard & Theorem~\ref{theo:min-dil-np-hard}\\
		$2$-dim. & complete graph & $\leq 2$ & $\mathcal{O}(n^3\log n)$ & \Cref{theo:2D-alg-complete-graph}\\
		$2$-dim. & plane & minimum & Min.~Dil.~Triangulation & Observation~\ref{observation:minimum_dilation_triangulation}\\
		$2$-dim. convex & plane & $\mathcal{O}(1)$ & $\mathcal{O}(n \log n )$ & \Cref{theo:alg-2D-greedy-convex} \\
		$1$-dim. & \textbf{--} & $1$ & $\mathcal{O}(n)$ & \Cref{theo:1D-1-spanner}\\
		$1$-dim. & $2$-page-plane & $\leq 2$ & $\mathcal{O}(n)$ & \Cref{theo:1D-2PPB-2-spanner} \\
		$1$-dim. & $1$-page-plane & minimum & $\mathcal{O}(n^7)$ & \Cref{theo:opt-alg-1PPB}\\
		$1$-dim. & $1$-page-plane & $\leq 5$ & $\mathcal{O}(n \log n)$ & \Cref{theo:alg-greedy}\\
	\end{tabular}    
\end{table}

%%%%%%%%%%%%%%%%%%%%%%%%%%%%%%%
\section{One-Dimensional Point Sets}\label{sec:1D}
%After showing that in general finding a minimum oriented spanner is NP-hard, we now restrict on one-dimensional point sets. This can be visualised by a number of points in the euclidean plane, which all have the same value in the $y$-coordinate. In this paper, one-dimensional point sets are numbered in increasing order, i.e.\ $p_1$ is the leftmost point and $p_n$ the rightmost point.

We first focus on points in one dimension. We will always draw points on a horizontal line with the minimum point leftmost, and the maximum point rightmost.
We observe that in one dimension only the dilation of a linear number of candidate pairs needs to be checked.

\begin{lemma}
	\label{theo:odil-1D-i+1-i+2-i+3}
	Let $P$ be a one-dimensional point set of $n$ points. The oriented dilation $t$ of an oriented %$t$-spanner 
	graph $G$ on $P$ is
	\[ t = \max\{\odil(p_i,p_{i+1}),\odil(p_j,p_{j+2}),\odil(p_k,p_{k+3})\} \]
	with $ 1\leq i\leq n-1$, $1\leq j\leq n-2$ and $1\leq k\leq n-3$.
\end{lemma}
\begin{proof}
	Let $p_l, p_r$ be any two distinct points in $P$ such that $l<r$.
	
	As the cases $r=l+1$, $r={l+2}$ and  $r={l+3}$ are given by the Lemma's statement, the statement is proven for $r<l+4$  or $l> n-3$.
	
	For $r\geq l+4$  and $l\leq  n-3$, we will show 
	\[ \odil(p_l, p_r) \leq \max\{\odil(p_j,p_{j+2}),\odil(p_k,p_{k+3})\}\]
	with %$ 1\leq i\leq n-1$,
	$1\leq j\leq n-2$ and $1\leq k\leq n-3$. We will induct on $r-l$. 
	
	%     We will show that for any two points $p_k,p_m\in P$ with $k<m$ and $k \leq n-3$, it holds that $\odil(p_k, p_m) \leq \max\{\odil(p_i,p_{i+2}),\odil(p_j,p_{j+3})\mid 1\leq i\leq n-2, 1\leq j\leq n-3\}.$

	%    Since $t\leq \max \{\odil(p_i,p_j)\mid p_i,p_j\in P\}$, we bound $\odil(p_i,p_j)$ for any two distinct points $p_i,p_j$ with $i<j$ and $i\leq n-3$. This is shown by induction on .
	
	%    For a one-dimensional point set $P$, the length of the optimal cycle $\Delta(p_k,p_{m})$ of any pair of points $p_k,p_{m}\in P$ with $m\geq k+2$ is two times the distance between these points, i.e.\ $|{\Delta(p_k,p_{m})}|={2\cdot\left(p_m- p_k\right)}$.
	
	%  %I.A.
	% For $m=k+1$, the third point of the optimal cycle $\Delta(p_k,p_{k+1})$ is $p_{k-1}$ or $p_{k+2}$. 
	% It holds that $|C_G(p_{k},p_{k+1})|\leq |C_G(p_{k},p_{k+2})|$ and $|C_G(p_{k},p_{k+1})|\leq |C_G(p_{k-1},p_{k+1})|$. 
	% If $\Delta(p_k,p_{k+1})$ contains $p_{k+2}$,  the dilation is bounded by
	% \begin{align*}
		%     \odil(p_k,p_{k+1})&=\frac{|C_G(p_{k},p_{k+1})|}{|\Delta(p_k,p_{k+1})|}=\frac{|C_G(p_{k},p_{k+1})|}{\left(p_{k+2}-p_{k}\right)\cdot 2}\\
		%     &\leq \frac{|C_G(p_{k},p_{k+2})|}{|\Delta(p_k,p_{k+2})|}=\odil(p_k,p_{k+2}).
		% \end{align*}
	
	% Otherwise, it holds that $\odil(p_k,p_{k+1})\leq \odil(p_{k-1},p_{k+1})$.    
	%  As the cases $m=k+1$, $m={k+2}$ and  $m={k+3}$ are given by the Lemma's statement, the statement is proven for $m<k+4$.
	
	%I.S.
	%   We consider $m\geq k+4$.
	%$C_G(p_k,p_m)$ is equal to or shorter than the union of $C_G(p_k,p_{k+2})$ and $C_G(p_{k+2},p_m)$ (see  \Cref{fig:od-i-i+3-split}). 
	Since the union of $C_G(p_l,p_{l+2})$ and $C_G(p_{l+2},p_r)$ forms a closed walk containing $p_l$ and $p_r$, the length of their shortest closed walk is bounded by 
	\[|C_G(p_l,p_r)| \leq |C_G(p_l,p_{l+2})|+ |C_G(p_{l+2},p_r)|.\]
	Further, for  $r\geq l+4$, the length of the optimal cycle $\Delta(p_l,p_{r})$ is two times the distance between these points, i.e.\ $|{\Delta(p_l,p_{r})}|={2\cdot\left(p_r- p_l\right)}$.
	This permits 
	\begin{align*}
		\odil(p_l,p_r)~&{=}~\frac{|C_G(p_l,p_r)|}{(p_l-p_r)\cdot 2} \leq \frac{|C_G(p_l,p_{l+2})|+ |C_G(p_{l+2},p_r)|}{(p_{r}-p_{l+2}+ p_{l+2}-p_l )\cdot2}\\  
		&{\leq}~  \max\Bigg\{\frac{|C_G(p_l,p_{l+2})|}{|\Delta(p_l,p_{l+2})|}, \frac{|C_G(p_{l+2},p_r)|}{|\Delta(p_{l+2},p_r)|}\Bigg\}\\
		&\hspace{3.5mm}= \max\{\odil(p_l,p_{l+2}),\odil(p_{l+2},p_r)\}.
	\end{align*}
	Recursively, this bounds the dilation of any tuple $p_l,p_r$ with $r\geq l+4$. 
\end{proof}
% \begin{figure}[ht]
	%     \centering
	%     \begin{tikzpicture}[graph, directed, every node/.append style={minimum width= 3em}]
		%         \node[] (0) {};
		%         \node[right= of 0] (1) {$p_k$};
		%         \node[right= of 1] (2) {$p_{k+1}$};
		%         \node[right= of 2] (3) {$p_{k+2}$};
		%         \node[right= of 3] (4) {};
		%         \node[right= of 4] (5) {$p_m$};
		%         \node[right= of 5] (6) {};
		
		%         \draw (0) edge[dotted] (1);
		%         \draw (1) edge[] (2);
		%         \draw (2) edge[] (3);
		%         \draw (3) edge[dotted] (4);
		%         \draw (4) edge[dotted] (5);
		%          \draw (5) edge[dotted] (6);
		
		%         \path[every edge/.append style={bend right, firstColor}]
		%         (6) edge (0)
		%         ;
		
		%         \path[every edge/.append style={bend right, secondColor}]
		%         (4) edge (2)
		%         (2) edge (0)
		%         ;
		
		%         \path[every edge/.append style={bend right, thirdColor}]
		%         (6) edge (2)
		%         ;
		%     \end{tikzpicture}
	%     \caption{Example: For $m\geq k+4$, the union of \textcolor{secondColor}{$C_G(p_k,p_{k+2})$} (\secondColor) and \textcolor{thirdColor}{$C_G(p_{k+2},p_m)$} (\thirdColor) covers  \textcolor{firstColor}{$C_G(p_k,p_m)$}  (\firstColor); for clarity we only coloured the back edges of the shortest closed walks.  This shows a $1$-PPB spanner, but since no step of the proof uses planarity, \Cref{theo:odil-1D-i+1-i+2-i+3} holds for any oriented graph.}
	%     \label{fig:od-i-i+3-split}
	% \end{figure}

This observation directly leads to an oriented $1$-spanner with $3n-6$ edges for every one-dimensional point set.% Note, that this graph is even sparse.

\begin{corollary}[oriented $1$-spanner]\label{theo:1D-1-spanner}
	For every one-dimensional point set $P$, $G=(P,E)$ with 
	\[E=\{(p_i,p_{i+1}),(p_{j+2},p_j),(p_{k+3},p_k)\mid 1\leq i\leq n-1, 1\leq j\leq n-2, 1\leq k\leq n-3\}\]
	is an oriented $1$-spanner on $P$ (see \Cref{fig:one-spanner-stack-triangulation}). %{Every $1$-spanner for $P$ contains $E$.}
\end{corollary}

\begin{figure}
	\begin{center}
		\begin{tikzpicture}[graph, directed, scale= 1.5]
			\node (v1) at (-4,0) {$p_1$};
			\node (v2) at (-3,0) {$p_2$};
			\node (v3) at (-2,0) {$p_3$};
			\node (v4) at (-1,0) {$p_4$};
			\node (v5) at (0,0) {$p_5$};
			\node (v6) at (1,0) {$p_6$};
			\node (v7) at (2,0) {$p_7$};
			\node (v8) at (3,0) {$p_8$};
			\draw  (v1) edge (v2);
			\draw  (v2) edge (v3);
			\draw  (v3) edge (v4);
			\draw  (v4) edge (v5);
			\draw  (v5) edge (v6);
			\draw  (v6) edge (v7);
			\draw  (v7) edge (v8);
			\draw [bend right] (v8) edge (v6);
			\draw [bend right] (v7) edge (v5);
			\draw [bend right] (v6) edge (v4);
			\draw [bend right] (v5) edge (v3);
			\draw [bend right] (v4) edge (v2);
			\draw [bend right] (v3) edge (v1);
			\draw [bend left] (v8) edge (v5);
			\draw [bend left] (v7) edge (v4);
			\draw [bend left] (v6) edge (v3);
			\draw [bend left] (v5) edge (v2);
			\draw [bend left] (v4) edge (v1);
		\end{tikzpicture}\\[3mm]
		
		\begin{tikzpicture}[graph, directed]
			\node (v1) at (-4,-3) {$p_1$};
			\node (v2) at (-4,3) {$p_2$};
			\node (v3) at (3,-.5) {$p_3$};
			\node (v4) at (-3,-2) {$p_4$};
			\node (v5) at (-2,1) {$p_5$};
			\node (v6) at (01,-0.5) {$p_6$};
			\node (v7) at (-2,-1) {$p_7$};
			\node (v8) at (-1,0) {$p_8$};
			\draw  (v1) edge (v2);
			\draw  (v2) edge (v3);
			\draw  (v3) edge (v4);
			\draw  (v4) edge (v5);
			\draw  (v5) edge (v6);
			\draw  (v6) edge (v7);
			\draw  (v7) edge (v8);
			\draw(v8) edge (v6);
			\draw(v7) edge (v5);
			\draw(v6) edge (v4);
			\draw(v5) edge (v3);
			\draw (v4) edge (v2);
			\draw (v3) edge (v1);
			\draw (v8) edge (v5);
			\draw  (v7) edge (v4);
			\draw (v6) edge (v3);
			\draw (v5) edge (v2);
			\draw (v4) edge (v1);
		\end{tikzpicture}
	\end{center}
	\caption{One-dimensional oriented $1$-spanner and its plane embedding in the Euclidean space}
	\label{fig:one-spanner-stack-triangulation}
\end{figure}

%The following still needs to be fixed (I at least replaced "not planar" by "not having a nice plane embedding" to make it correct,  but this can't stay like that. Instead we should discuss: The 1-spanner, if we lay  out the points on the $x$-axis and draw the edges as $x$-monotone arcs, will have a linear number of crossings [We don't need to prove this]. The graph is actually planar. This can be easily seen by induction: The first three points form a triangle, the forth we can place inside and add edges to the first three. In this embedding, in particular the 2nd to fourth point form a triangular face. Generally can add $p_{i+3}$ in the triangle formed by $p_i, p_{i+1}$ and $p_{i+2}$, obtaining a stack triangulation. However, we can't even embed this graph (for more than 4 points) in such a way that the smallest and largest point are on the outer face: If we could, since they do not share an edge, we could add it while maintaining planarity. But the resulting graph has $3n-5$ edges and therefore cannot be planar.}

In two dimensions, plane (straight-line) spanners are of particular interest~\cite{Bose.2013}. The natural one-dimensional analogue to plane straight-line graphs are one- and two-page book embeddings~\cite{BekosGR16.bookembeddings,ChungLR87.bookembeddings,DujmovicW04.bookembeddings,Yannakakis89.bookembeddings}.

A \emph{one-page book embedding} of a graph corresponds to an embedding of the vertices as points on a line with the edges drawn without crossings as circular arcs above the line.  In a \emph{two-page book embedding} an edge may be drawn as an arc above or below the line. In such a (one- or two-page) book embedding, for consecutive points on the line, we may draw their edge straight on the line.
We call a graph \emph{one-page plane} (respectively \emph{two-page plane}) if it has a one-page (respectively two-page) book embedding.

In particular, one-page plane graphs are outerplanar graphs and correspond to plane straight-line graphs if we embed the points on a (slightly) convex arc instead of on a line. Two-page plane graphs are a subclass of planar graphs, %while not all three-page plane graphs are planar and 
while any planar graph has a four-page book embedding~\cite{Yannakakis89.bookembeddings}.

As we argue next, the graph given in \Cref{theo:1D-1-spanner} is not two-page plane with the given one-dimensional embedding (and therefore also not one-page plane). This follows from a stronger statement:

\begin{observation}\label{theo:1D-1-spanner-not-page-plane}
The graph given in \Cref{theo:1D-1-spanner} has no plane embedding with the first and last point on the outer face. 
\end{observation}
\begin{proof}
Suppose the graph would have such an embedding. The graph has $3n-6$ edges, but no edge between the first and the last point for $n>4$. Thus, we could add this edge while maintaining planarity, which contradicts the fact that a planar graph has at most $3n-6$ edges. 
%Interestingly, the $1$-spanner is planar: the graph corresponds to the stack triangulation that we obtain by adding the points from left to right, while always adding the next point into the triangle of the previous three points.
\end{proof}

Therefore, the graph given in \Cref{theo:1D-1-spanner} is not two-page plane (and thus not one-page plane). 
So, for a point set $P$ with $|P|>3$, every two-page plane graph $G=(P,E)$ contains a tuple $p_i,p_{i+2}\in P$ where $(p_{i+2}, p_i)\notin E$ or $p_i,p_{i+3}\in P$  where $(p_{i+3},p_i)\notin E$. Combining this observation with \Cref{theo:odil-1D-i+1-i+2-i+3}, we follow:

\begin{corollary}\label{theo:bound-no-1-1d-1-2-page-plane}
There is no two-page plane (and therefore also not one-page plane) $1$-spanner for any point set $P$ with $|P|>3$.
\end{corollary}

Interestingly, the graph of \Cref{theo:1D-1-spanner} is planar. We construct a stack triangulation by adding points from left to right. The first three points form a triangle. Then, we inductively add the next point into the triangle formed by the previous three points. See \Cref{fig:one-spanner-stack-triangulation}.

\subsection{Two-Page Plane Spanners}\label{sec:1D-2PPB}
Due to \Cref{theo:bound-no-1-1d-1-2-page-plane}, no two-page plane $1$-spanner can exist. 
However, we can give a two-page plane $2$-spanner for every one-dimensional point set:

\begin{proposition}[two-page plane $2$-spanner]
\label{theo:1D-2PPB-2-spanner}
For every one-dimensional point set $P$, the graph $G=(P,E)$ with 
$E=\{(p_i,p_{i+1}),(p_{j+2},p_j)\mid 1\leq i\leq n-1, 1\leq j\leq n-2\}$ (see \Cref{fig:1D-2PPB-2-spanner})
is a two-page plane oriented $2$-spanner for $P$.
\end{proposition}
\begin{figure}[ht]
\centering
\begin{tikzpicture}[graph, directed, every node/.append style={minimum width=3em}, node distance=2.4em]
	\node[] (i) {$p_{i}$};
	\node[left= of i] (i-1) {$p_{i-1}$};
	\node[right= of i] (i+1) {$p_{i+1}$};
	\node[right= of i+1] (i+2) {$p_{i+2}$};
	\node[right= of i+2] (i+3) {$p_{i+3}$};
	\node[right= of i+3] (i+4) {$p_{i+4}$};
	\node[right= of i+4, graph/plain] (i+5) {};
	\node[left= of i-1, graph/plain] (i-2) {};
	
	\path[]
	(i-2) edge[dotted] (i-1)
	(i-1) edge (i)
	(i) edge (i+1)
	(i+1) edge (i+2)
	(i+2) edge (i+3)
	(i+3) edge (i+4)
	(i+4) edge[dotted] (i+5)
	(i+1) edge [bend left] (i-1)
	(i) edge [bend right, dotted] (i-2)
	(i+2) edge [bend right] (i)
	(i+3) edge [bend left] (i+1)
	(i+4) edge [bend right] (i+2)
	(i+5) edge [bend left, dotted] (i+3)
	;
\end{tikzpicture}
\caption{Part of $G=(P,E)$ with  $E=\{(p_i,p_{i+1}),(p_{j+2},p_j)\mid 1\leq i\leq n-1, 1\leq j\leq n-2\}$.}
\label{fig:1D-2PPB-2-spanner}
\end{figure}
\begin{proof}
Since $(p_{i+2},p_i)\in E$, it holds that $\odil(p_i,p_{i+2})=1$ for each tuple $p_i,p_{i+2}\in P$ and also $\odil(p_i,p_{i+1})=1$ for each tuple $p_i,p_{i+1}\in P$.
Due to \Cref{theo:odil-1D-i+1-i+2-i+3}, if we prove $\odil(p_i,p_{i+3})\leq 2$ for any tuple $p_i,p_{i+3}\in P$, then $G$ is a $2$-spanner for $P$.

We consider the closed walk 
\[C=  \{(p_{i+3},p_{i+1}),(p_{i+1},p_{i+2}),(p_{i+2},p_{i})\}\cup\{(p_j,p_{j+1})\mid i\leq j\leq i+2\}.\]

Since $C$ contains $p_{i}$ and $p_{i+3}$, it holds 
$|C_G(p_i,p_{i+3})|\leq |C|$ and therefore     
\begin{align*}
	\odil(p_i,p_{i+3})&=\frac{|C_G(p_i,p_{i+3})|}{|\Delta(p_i,p_{i+3})|} \leq \frac{|C|}{|\Delta(p_i,p_{i+3})|}&\\
	&= \frac{p_{i+3}-p_{i+1}+p_{i+2}-p_{i+1}+p_{i+2}-p_{i}+p_{i+3}-p_i}{(p_{i+3}-p_i)\cdot 2}&\\
	&=\frac{(p_{i+3}-p_{i+1})\cdot 2+(p_{i+1}-p_{i})\cdot 2+(p_{i+2}-p_{i+1})\cdot 2}{(p_{i+3}-p_i)\cdot 2}&\\
	&=\frac{p_{i+3}-p_{i}+p_{i+2}-p_{i+1}}{p_{i+3}-p_i} =1+\frac{p_{i+2}-p_{i+1}}{p_{i+3}-p_i} \leq  2. &\tag*{\qedhere}
\end{align*}    
%    There are three options for the shortest closed walk  $C_G(p_i,p_{i+3})$.
% \begin{enumerate}[label=\roman*)]
	%   \item If $C_G(p_i,p_{i+3})= \{(p_{i+3},p_{i+1}),(p_{i+1},p_{i+2}),(p_{i+2},p_{i})\}\cup\{(p_j,p_{j+1})\mid i\leq j\leq i+2\}$, then 
	%   \begin{align*}
		%        \odil(p_i,p_{i+3})&=\frac{|C_G(p_i,p_{i+3})|}{(p_{i+3}-p_i)\cdot 2}\\
		%        &=\frac{p_{i+3}-p_{i+1}+p_{i+2}-p_{i+1}+p_{i+2}-p_{i}+p_{i+3}-p_i}{(p_{i+3}-p_i)\cdot 2}\\
		%        &=\frac{(p_{i+3}-p_i)\cdot 2+(p_{i+2}-p_{i+1})\cdot 2}{(p_{i+3}-p_i)\cdot 2}
		%        =1+\frac{p_{i+2}-p_{i+1}}{p_{i+3}-p_i} \leq  2.
		%    \end{align*}
	
	%    \item If $C_G(p_i,p_{i+3})= \{(p_{i+4},p_{i+2}),(p_{i+2},p_{i})\}\cup\{(p_j,p_{j+1})\mid i\leq j\leq i+3\}$, then
	%    \begin{align*}
		%        \odil(p_i,p_{i+3})&=\frac{|C_G(p_i,p_{i+3})|}{(p_{i+3}-p_i)\cdot 2}=\frac{(p_{i+4}-p_i)\cdot 2}{(p_{i+3}-p_i)\cdot 2}
		%        =1+\frac{p_{i+4}-p_{i+3}}{p_{i+3}-p_i} \\
		%        &\overset{\substack{\text{if } p_{i+4}-p_{i+3}\\ \leq p_{i+3}-p_i}}{\leq } 2.
		%    \end{align*}
	%    If $p_{i+4}-p_{i+3}> p_{i+3}-p_i$, then case~i) is a shorter closed walk for the tuple $p_i,p_{i+3}$. 
	%    \item If $C_G(p_i,p_{i+3})= \{(p_{i+3},p_{i+1}),(p_{i+1},p_{i-1})\}\cup\{(p_j,p_{j+1})\mid i-1\leq j\leq i+2\}$, then the case is analogous to  case~ii).
	%\qedhere
	% \end{enumerate}
	\end{proof}

	%%%%%%%%%%%%%%%%%%%%%%%%%%%%%%%
	\subsection{One-Page Plane Spanners}\label{sec:1D-1PPB}
	
	The $2$-spanner given by Proposition \ref{theo:1D-2PPB-2-spanner} is two-page plane, but not one-page plane. As noted above, one-page plane graphs on one-dimensional point sets correspond to plane straight-line graphs if we interpret the point set as being convex. We thus place particular focus on one-page plane graphs, since they are not only of interest in their own right but also aid us in finding oriented plane spanners in two dimensions.
	
	%But this is not the only reason to consider one-page plane spanners. 
	%As mentioned before, finding optimal spanners seems to be hard for two-dimensional point sets. As a restriction for that, one could consider the natural generalisation of the one-dimensional case, which would be point sets in convex position. Graphically, by splitting the convex point set and unbending it, we would create a one-dimensional point set.
	%Regarding a plane spanner on a convex point set would only allow straight lines from one point to another. The equivalent on the split and unbent point set is a one-page plane spanner.
	
	By \emph{maximal one-page plane}, we mean a one-page plane graph $G=(P,E)$ such that for every edge $e\notin E$, the graph $G'=(P, E \cup \{e\})$ is not  one-page plane. % Examples for maximal one-page plane graphs are shown in \Cref{fig:1-sided-5points-all-graphs}.
	We call the edge set $\{(p_i,p_{i+1})\mid 1\leq i\leq n-1\}$ the \emph{baseline}. A directed edge $(p_j,p_i)$ with \ $i<j$ is a \emph{back edge}.
	
	We refer to an oriented one-page plane graph that includes a baseline and all other edges are back edges as a one-page-plane-baseline graph, for short \emph{$1$-PPB graph}.
	\Cref{theo:baseline} shows that a $1$-PPB graph has at most the same dilation as an oriented graph with the same edge set but another orientation. 
	For the remainder of this paper, we will focus exclusively on $1$-PPB graphs rather than general one-page plane graphs. This restriction is justified by \Cref{theo:baseline}, which indicates that it does not affect the minimum attainable dilation. Consequently, the construction of  one-page plane spanners is simplified.  
	Notably, computing a minimum dilation of a $1$-PPB graph yields a minimum dilation one-page plane graph (compare to \Cref{theo:opt-alg-1PPB}).
	
	\begin{lemma}
\label{theo:baseline}
Let $G=(P,E)$ be a one-page plane oriented $t$-spanner for a one-dimensional point set $P$. 
Consider a $1$-PPB graph $G'=(P,E')$, where $E'$ contains the baseline edges and for every edge in $E$ the same edge but oriented according to the conditions of a $1$-PPB graph. Then the dilation of $G'$ is at most $t$. 
%Let $E'$ be an edge set where the edges are incident to the same vertices as $E$, but the orientation 
%may be different. If we orient such that $G'=(P,E')$ is a $1$-PPB graph, then the dilation of $G'$ is at most $t$. 
%i.e.\ $(p,p')\in E$ implies  $(p,p')\in E'$  or  $(p',p)\in E'$.
% If the edges in $E'$ are oriented such that
% \begin{itemize}
	%     \item $E'$ includes a baseline, i.e.\ $\{(p_i,p_{i+1})\mid 1\leq i\leq n-1 \}\subset E'$, and
	%     \item all other edges in $E'$ are back edges, i.e.\ $(p_i,p_j)\notin E$ for all $p_i,p_j\in P$ with $j>i+1$,
	% \end{itemize}  
%it holds that $t'\leq t$.
\end{lemma}
\begin{proof}
Without loss of generality we assume that for all $i$ with $1 \leq i \leq n-1$ either $(p_i, p_{i+1})$ or $(p_{i+1}, p_i)$ is in $E$, since adding these would only decrease $t$ while $E'$ contains the baseline anyway.
%For any $p_i\in P$, if neither  $(p_i, p_{i+1})$ nor $(p_{i+1}, p_i)$ is in $E$, adding the baseline edge $(p_i, p_{i+1})$ might only decrease the dilation of $G'$.
Let $G^u=(P,E^u)$ be the underlying undirected graph of $G$. 
Let $p,p'$ be two points in $P$. 
Let $C_G(p,p')$ be the shortest closed walk in $G$ containing $p$ and $p'$. Let $p_l$ be the leftmost and $p_r$ be the rightmost point of $C_G(p,p')$, respectively.
Consider the edges in $G^u$ corresponding to $C_G(p,p')$, that is $E_C^u=\{\{p_i,p_j\} \in E^u \mid (p_i, p_j) \in C_G(p,p') \text{ or } (p_j,p_i) \in C_G(p,p')\}$. 

We define the \emph{inclusion maximal edges} in this set as those edges, which are not covered by longer edges, that is $\imax(E_C^u) = \{ \{p_i,p_j\} \in E_C^u \mid p_i', p_j' \in P \text{ with } i' \leq i < j \leq j':~ \{p_i', p_j'\} \notin E_C^u \}$.
For this set, we can state that
\begin{itemize}
	\item[i)] For all $\{p_i, p_j\} \in \imax(E_C^u)$ we have $j > i+1$, because an edge $\{p_i, p_{i+1}\}$ cannot be inclusion maximal, since the interval between $p_i$ and $p_{i+1}$ has to be covered by at least two edges to allow a connected graph.
	\item[ii)] The edges in $\imax(E_C^u)$ partition %cover disjoint 
	an interval $[p_\ell, p_r]$, as $G$ is one-page plane and thus $C_G(p,p')$ would be disconnected otherwise.
\end{itemize}
By using only the inclusion maximal and baseline edges, 
it follows that the orientation of $E_C^u$ which is given by $G'$ includes a closed walk of length $2 \cdot (p_r - p_\ell) $ and further: 
$$|C_{G'}(p,p')| \leq 2 \cdot (p_r - p_\ell ) \leq |C_G(p,p')|\;.$$
Thus the oriented dilation of $p$ and $p'$ in the $1$-PPB $G'$ is smaller or equals their oriented dilation in $G$.
\end{proof}
\begin{lemma}
\label{theo:odil-1D-1sided}
Let $P$ be a one-dimensional point set of $n$ points. The oriented dilation $t$ of a $1$-PPB graph for $P$ is 
\[t= \max\{\odil(p_i,p_{i+2})\mid 1\leq i\leq n-2 \}.\]
\end{lemma}
\begin{proof}
Let $G$ be a $1$-PPB spanner. We show that for any point $p_j\in P$ it holds that
\begin{align*}
	\odil(p_j,p_{j+1})\leq {\max} \{\odil(p_j,p_{j+2}),\odil(p_{j-1},p_{j+1})\} & \text{ and }\\
	\odil(p_j,p_{j+3})\leq {\max} \{\odil(p_j,p_{j+2}),\odil(p_{j+1},p_{j+3})\}.&
\end{align*}
By \Cref{theo:odil-1D-i+1-i+2-i+3}, this completes the proof of  this lemma.

First, we show $\odil(p_j,p_{j+1})\leq {\max} \{\odil(p_j,p_{j+2}),\odil(p_{j-1},p_{j+1})\}$.

The third point of the optimal cycle $\Delta(p_j,p_{j+1})$ is $p_{j-1}$ or $p_{j+2}$. 
For graphs with baseline, it holds that $|C_G(p_{j},p_{j+1})|\leq |C_G(p_{j},p_{j+2})|$ and $|C_G(p_{j},p_{j+1})|\leq |C_G(p_{j-1},p_{j+1})|$. 
If $\Delta(p_j,p_{j+1})$ contains $p_{j+2}$,  the dilation is bounded by
\begin{align*}
	\odil(p_j,p_{j+1})&=\frac{|C_G(p_{j},p_{j+1})|}{|\Delta(p_j,p_{j+1})|}=\frac{|C_G(p_{j},p_{j+1})|}{\left(p_{j+2}-p_{j}\right)\cdot 2}\\
	&\leq \frac{|C_G(p_{j},p_{j+2})|}{|\Delta(p_j,p_{j+2})|}=\odil(p_j,p_{j+2}).
\end{align*}   
Otherwise, we can show analogously that $\odil(p_k,p_{k+1})\leq \odil(p_{k-1},p_{k+1})$.

Now, we show  $\odil(p_j,p_{j+3})\leq {\max} \{\odil(p_j,p_{j+2}),\odil(p_{j+1},p_{j+3})\}$. For graphs with baseline,  we have $|C_G(p_{j+1},p_{j+3})| \leq |C_G(p_{j},p_{j+3})|$.

If $|C_G(p_j,p_{j+3})|=|C_G(p_{j+1},p_{j+3})|$, it holds that
\begin{align*}
	\odil(p_j,p_{j+3}) &=\frac{|C_G(p_{j},p_{j+3})|}{| \Delta(p_{j},p_{j+3})|}= \frac{|C_G(p_{j+1},p_{j+3})|}{(p_{j+3}-p_{j})\cdot 2}\\
	&\leq \frac{|C_G(p_{j+1},p_{j+3})|}{(p_{j+3}-p_{j+1})\cdot 2} =\odil(p_{j+1},p_{j+3}).
\end{align*}

If $|C_G(p_{j},p_{j+3})|> |C_G(p_{j+1},p_{j+3})|$, the leftmost point of $C_G(p_{j+1},p_{j+3})$ is $p_{j+1}$. This implies $|C_G(p_j,p_{j+2})|=|C_G(p_{j},p_{j+3})|$ and therefore 
\begin{align*}
	\odil(p_j,p_{j+3}) &=\frac{|C_G(p_{j},p_{j+3})|}{| \Delta(p_{j},p_{j+3})|}= \frac{|C_G(p_{j},p_{j+2})|}{(p_{j+3}-p_{j})\cdot 2} &\\
	&\leq \frac{|C_G(p_{j},p_{j+2})|}{(p_{j+2}-p_{j})\cdot 2} =\odil(p_{j},p_{j+2}).\tag*{\qedhere}
\end{align*}
\end{proof}

\Cref{theo:odil-1D-1sided} holds for every $1$-PPB graph, even if it is not maximal. \Cref{theo:1D-2PPB-2-spanner} shows that \Cref{theo:odil-1D-1sided} does not hold for non-one-page plane graphs.

%%%%%%%%%%%%%%%   Bounds Odil 1-sided  %%%%%%%%%%%%%%%%

Due to planarity, for a point set $P$ with $|P|>3$, every one-page plane graph $G$ contains a tuple $p_i,p_{i+2}\in P$ where $C_G(p_i,p_{i+2})\neq \Delta(p_i,p_{i+2})$. So, there is no one-page plane $1$-spanner for any point set $P$ with
$|P|>3$ (\Cref{theo:bound-no-1-1d-1-2-page-plane}). 

However, there are point sets with a one-page plane almost $1$-spanner.
Let $G=(P,E)$ be a graph with $|P|=5$, $E=\{(p_{i},p_{i+1})\mid 1\leq i\leq n-1\}\cup \{(p_3,p_1),(p_5,p_3)\}$ and the distances $p_2-p_1=p_5-p_4=\varepsilon$ and $p_3-p_2=p_4-p_3=1$ (see \Cref{fig:1D-1PPB-almost-1-spanner}). It holds that $t= \odil(p_2,p_4)=\frac{2+2\varepsilon}{2}$. For an arbitrary small $\varepsilon$, $G$ is a one-page plane almost $1$-spanner. 

\begin{figure}[ht]
\centering
\begin{tikzpicture}[graph, directed,scale=1.2]
	\foreach \a in {1,...,5}{
		\node[] (\a) at (\a,0) {$p_\a$} ;
	}
	
	\draw (1) edge (2);
	\draw (2) edge (3);
	\draw (3) edge (4);
	\draw (4) edge (5);
	
	\path[every edge/.append style={bend right=35}]
	(3) edge (1)  
	(5) edge (3)
	% (5) edge (1) 
	;
	
	\path[]
	($(1.south)-(0,.3)$) edge[|-]  node[below] {$\varepsilon$} ($(2.south)-(0,.3)$)
	($(2.south)-(0,.3)$) edge[|-]  node[below] {$1$} ($(3.south)-(0,.3)$)
	($(3.south)-(0,.3)$) edge[|-]  node[below] {$1$} ($(4.south)-(0,.3)$)
	($(5.south)-(0,.3)$) edge[|-|]  node[below] {$\varepsilon$} ($(4.south)-(0,.3)$)         
	;
\end{tikzpicture}
\caption{$1$-PPB almost $1$-spanner}
\label{fig:1D-1PPB-almost-1-spanner}
\end{figure}

% Further, it holds:

\begin{observation}
\label{theo:bound-2-1d-1sided}
There are one-dimensional point sets where no one-page plane oriented $t$-spanner exists for $t < 2$.
\end{observation}
\begin{proof}
We show that $t\geq 2$ holds for every $t$-spanner of the point set $P$ with $|P|=5$ where each consecutive pair of points has distance $1$.
Since a $1$-PPB spanner has a smaller or equal dilation than any one-page plane spanner with the same edge set but another orientation (\Cref{theo:baseline}) and maximality decreases dilation, it suffices to look at all maximal $1$-PPB graphs for $P$. These are listed in \Cref{fig:1-sided-5points-all-graphs}. %Since each consecutive pair of points has distance $1$, it holds $|\Delta(p_i,p_{i+2})|=4$. 
For every maximal $1$-PPB graph $G$ for $P$, there is a tuple $p_i,p_{i+2}\in P$ where $p_1$ is the leftmost point in $C_G(p_i,p_{i+2})$ and $p_5$ the rightmost point. %, i.e.\ $C_G(p_i,p_{i+2})=(p_1,p_2),\ldots,(p_4,p_5),(p_5,p_1)$. %In \Cref{fig:1-sided-5points-all-graphs}a) that is $p_2,p_4$, in b,c) $p_3,p_5$ and in d,e) $p_1,p_3$. 
Combining these observations with \Cref{theo:odil-1D-1sided}, gives
\[t =  \underset{1\leq i \leq n-2}{\max} \frac{|C_G(p_i,p_{i+2})|}{|\Delta(p_i,p_{i+2})|} \geq \frac{\left(p_5-p_1\right)\cdot 2}{4}=\frac{8}{4}=2\]
for the dilation of every one-page plane oriented spanner for $P$.
\end{proof}

\begin{figure}[ht]
\centering
\begin{tikzpicture}[graph, directed]
	\node[] (1) {$p_1$};
	\node[right= of 1,rectangle, inner sep=5pt] (2) {$p_2$};
	\node[right= of 2] (3) {$p_3$};
	\node[right= of 3,rectangle, inner sep=5pt] (4) {$p_4$};
	\node[right= of 4] (5) {$p_5$};
	
	\draw (1) edge (2);
	\draw (2) edge (3);
	\draw (3) edge (4);
	\draw (4) edge (5);
	
	\path[every edge/.append style={bend right}]
	(3) edge (1)  
	(5) edge (3)
	(5) edge (1) 
	;
\end{tikzpicture}\\
\begin{tikzpicture}[graph, directed]
	\node[] (1) {$p_1$};
	\node[right= of 1] (2) {$p_2$};
	\node[right= of 2,rectangle, inner sep=5pt] (3) {$p_3$};
	\node[right= of 3] (4) {$p_4$};
	\node[right= of 4,rectangle, inner sep=5pt] (5) {$p_5$};
	
	\draw (1) edge (2);
	\draw (2) edge (3);
	\draw (3) edge (4);
	\draw (4) edge (5);
	
	\path[every edge/.append style={bend right}]
	(3) edge (1)
	(4) edge (1)
	(5) edge (1)
	;
\end{tikzpicture}\\
\begin{tikzpicture}[graph, directed]
	\node[] (1) {$p_1$};
	\node[right= of 1] (2) {$p_2$};
	\node[right= of 2,rectangle, inner sep=5pt] (3) {$p_3$};
	\node[right= of 3] (4) {$p_4$};
	\node[right= of 4,rectangle, inner sep=5pt] (5) {$p_5$};
	
	\draw (1) edge (2);
	\draw (2) edge (3);
	\draw (3) edge (4);
	\draw (4) edge (5);
	
	\path[every edge/.append style={bend right}]
	(4) edge (2)
	(4) edge (1)
	(5) edge (1)
	;
\end{tikzpicture}\\
\begin{tikzpicture}[graph, directed]
	\node[rectangle, inner sep=5pt] (1) {$p_1$};
	\node[right= of 1] (2) {$p_2$};
	\node[right= of 2,rectangle, inner sep=5pt] (3) {$p_3$};
	\node[right= of 3] (4) {$p_4$};
	\node[right= of 4] (5) {$p_5$};
	
	\draw (1) edge (2);
	\draw (2) edge (3);
	\draw (3) edge (4);
	\draw (4) edge (5);
	
	\path[every edge/.append style={bend right}]
	(4) edge  (2)
	(5) edge  (2)
	(5) edge  (1)
	;
\end{tikzpicture}\\
\begin{tikzpicture}[graph, directed]
	\node[rectangle, inner sep=5pt] (1) {$p_1$};
	\node[right= of 1] (2) {$p_2$};
	\node[right= of 2,rectangle, inner sep=5pt] (3) {$p_3$};
	\node[right= of 3] (4) {$p_4$};
	\node[right= of 4] (5) {$p_5$};
	
	\draw (1) edge (2);
	\draw (2) edge (3);
	\draw (3) edge (4);
	\draw (4) edge (5);
	
	\path[every edge/.append style={bend right}]
	(5) edge  (3)
	(5) edge   (2)
	(5) edge  (1)
	;
	
	\foreach \a in {2,3,4}{
		\pgfmathtruncatemacro\b{\a-1}
		\draw ($(\a)-(0,.7)$) edge[-|]  node[below] {$1$} ($(\b)-(0,.7)$);
	}
	\draw ($(5)-(0,.7)$) edge[|-|]  node[below] {$1$} ($(4)-(0,.7)$);
\end{tikzpicture}
\caption{All maximal $1$-PPB spanners for a one-dimensional set of $5$ points. The worst-case pair is highlighted.}%, where each consecutive pair of points has distance $1$}
\label{fig:1-sided-5points-all-graphs}
\end{figure}

%%%%%%%%%%%%%%%%%%%%%%%%%%%%%%%

\subsubsection{Constant-Dilation One-Page Plane Spanner}

% We construct a $1$-PPB spanner by starting with the baseline and greedily adding back edges sorted by length from shortest to longest if they do not cross any of the edges already added. The following leads to a simple, efficient algorithm for constructing this graph (compare to \Cref{alg:greedy-1page-1dim}): The first edge that we need to add is the shortest edge between two points with exactly one point in between. Imagine deleting the point that was in between. Then again, we need to add the shortest edge with exactly one point in between and so on, until only two points are left. The resulting graph is an oriented greedy triangulation.
% By maintaining the points in a doubly linked list and the relevant distances in a priority queue, this leads to a runtime of $\bigo (n \log n)$.

A first approach to construct a one-page plane oriented spanner is orienting the greedy triangulation of a given point set. We will show that this leads to a maximal $1$-PPB $5$-spanner.

\subparagraph{Algorithm.} 
\Cref{alg:greedy-1page-1dim} computes a maximal $1$-PPB spanner by starting with the baseline and greedily adding back edges sorted increasing by length if they do not cross any of the edges already added. The first added edge is the shortest edge between two points with exactly one point in between. Imagine we delete the point that was in between. Then again, we add the shortest edge with exactly one point in between, and so on, until only two points are left. 
This deletion is implemented by a list $L$, which is initialised with the point set. We delete a point $p$ from $L$ when we add an edge that covers $p$. In order to quickly compute the shortest edge, for each point $p$ in $L$, the length of the shortest edge that covers $p$ is stored in a priority queue.
The resulting graph is an oriented greedy triangulation.\bigskip

\begin{algorithm}[t]
\caption{Greedy One-Page Plane $5$-Spanner}
\label{alg:greedy-1page-1dim}
\begin{algorithmic}%[1]
\REQUIRE one-dimensional point set $P=\{p_1,\ldots,p_n\}$ (numbered from left to right)
\ENSURE one-page plane oriented $5$-spanner for $P$
\STATE $E\leftarrow \{(p_{i},p_{i+1})\mid 1\leq i \leq n-1 \}$ \COMMENT{start with the baseline}
\STATE $L\leftarrow$ doubly linked list with the elements $p_1,\ldots,p_n$
\STATE $T\leftarrow$ empty priority queue\\
$T$ stores tuples $(p,v)$ where $p$ is a pointer to a point in $L$ and the value $v=L(p).\texttt{next}-L(p).\texttt{prev}$ (distance of $p$'s neighbours in $L$)  \COMMENT{$v$ is the length of the edge which covers $p$}\\
%  $T$ is sorted according to the values $v$
\FOR{$i=2$ \textbf{to} $n-1$ }
\STATE add $(p_i,\left|p_{i+1}-p_{i-1}\right|)$ to $T$ \COMMENT{initially all edges with $1$ point in between}
\ENDFOR
%  \STATE $E\leftarrow\emptyset$  \COMMENT{back edges}
\WHILE{$T\neq \emptyset$}
\LINECOMMENT{$(p_r,p_l)$ is the current shortest edge}
\STATE $p\leftarrow T.\texttt{extractMin()}$
\STATE $p_l\leftarrow L(p).\texttt{prev}$
\STATE $p_r\leftarrow L(p).\texttt{next}$
\STATE $E\leftarrow E\cup \{(p_r,p_l)\}$ 
\LINECOMMENT{$p$ is covered by $(p_r,p_l)$, i.e.\ due to planarity there will be no more edges incident to $p$}
\STATE delete $p$ in $L$ and $T$ 
\STATE update the values in $T$ of $p_l$ and $p_r$ 
\ENDWHILE
%       \STATE Sort  edges in $E'= \{(p_{j},p_{i})\mid  1\leq i+2 \leq j \leq n \}$ ascending by length  \COMMENT{possible back edges}
% \FORALL{$(p_{j},p_{i})\in E'$}
% \IF[preserve planarity]{$\nexists (p_{r},p_{l})\in E \text{ with ($l<i<r<j$) or ($i<l<j<r$)}$}
% \STATE	$E\leftarrow E \cup \{(p_j,p_i)\}$  \COMMENT{add  back edge $(p_{j},p_{i})$}
% \ENDIF
% \ENDFOR
\RETURN $G=(P,E)$ 
\end{algorithmic}
\end{algorithm}

To bound the dilation of the output graph, we need the concept of a \emph{blocker}. 

\begin{definition}[blocker]\label{def:greedy-blocker}
Let $E$ be the greedily computed edge set. Because the resulted graph $G=(P,E)$ is maximal, $(p_j,p_i)\notin E$ for $i+2 \leq j$ implies there is a shorter edge in $E$ which intersects with $(p_j,p_i)$. (The greedy algorithm added this edge first and discarded $(p_{j},p_i)$ in a later iteration of the loop.)  
For the shortest edge intersecting $(p_{j},p_i)$, we say it \emph{blocks} $(p_{j},p_i)$. The edge can \emph{be blocked by} $(p_{k},p_{m})$ with $k> j$ and $i<m<j$ or $(p_{m},p_{k'})$ with $k'<i$ and $i<m<j$ (see \Cref{fig:blockers-for-i-j}). 
\end{definition}

\begin{figure}[ht]
\centering
\begin{tikzpicture}[graph, every node/.append style={minimum size=2.8em}, directed]
\node[] (k') {$p_{k'}$};
\node[right= of k'] (i) {$p_{i}$};
\node[right= of i] (m) {$p_{m}$};
\node[right= of m] (j) {$p_{j}$};
\node[right= of j] (k) {$p_{k}$};

\path[]
(k') edge[dotted] (i)
(i) edge[dotted] (m)
(m) edge[dotted] (j)
(j) edge[dotted] (k)
(j) edge[bend right, dashed, firstColor] (i)
(k) edge[bend right] (m)
(m) edge[bend right] (k')
;
\end{tikzpicture}    
\caption{$(p_{j},p_i)$ can be blocked by $(p_{k},p_{m})$ or $(p_{m},p_{k'})$ with $i<m<j$, $k> j$ and $k'<i$}
\label{fig:blockers-for-i-j}
\end{figure}

\begin{theorem}[one-page plane $5$-spanner]
\label{theo:alg-greedy}
Given a one-dimensional point set $P$ of size $n$, a one-page plane oriented $5$-spanner can be constructed in $\bigo(n\log n)$ time.
\end{theorem}
\begin{proof}%\newcounter{labeledequotationcounter}%eqoutation lable: \circ
Let  $G=(P,E)$ be the maximal $1$-PPB spanner constructed by \Cref{alg:greedy-1page-1dim}. 
For every $p_i,p_{i+2}\in P$ with $i\leq n-2$, we show $\odil(p_i,p_{i+2})\leq 5$.  Due to \Cref{theo:odil-1D-1sided}, this guarantees that $G$ is a $5$-spanner. 

If $(p_{i+2},p_i)\in E$, it holds that $\odil(p_i,p_{i+2})=1$. 

For $(p_{i+2},p_i)\notin E$, let $\{b_1,p_{i+1}\}$  be the shortest edge that blocks $(p_{i+2},p_i)$. 
We show the case that $b_1$ is right of $p_{i+2}$ (and therefore $i\leq n-3$). The other case is analogous.
% If $(p_{i+1},p_{k'})$ is the blocker, the argument can be mirrored.
Due to maximality, $b_1$ is the leftmost point with a back edge to $p_{i+1}$. %Therefore,  $(b_1,p_{i+2})\in E$, if $b_1\neq p_{i+3}$.
Consider \Cref{fig:1D-alg-greedy-i-i+2-blocked}.
Since $(b_1,p_{i+1})$ blocks $(p_{i+2},p_i)$, it holds that
\begin{equation}\label{eq:greedy-choice-i+1-k}
b_1-p_{i+1}\leq p_{i+2}-p_{i}.
\end{equation} 

\begin{figure}[ht]
\centering
\begin{tikzpicture}[graph,directed, node distance=4em, every node/.append style={minimum width=3.5em}]
	\node[] (i) {$p_{i}$};
	\node[right= of i] (i+1) {$p_{i+1}$};
	\node[right= of i+1] (i+2) {$p_{i+2}$};
	\node[right= of i+2] (k) {$b_1$};
	
	\path[]
	(i) edge (i+1)
	(i+1) edge (i+2)
	(i+2) edge[dotted] (k)
	%  (k) edge[bend right] (i+2)
	(k) edge[bend right] (i+1)
	(i+2) edge[bend right, firstColor, dashed] (i)
	;
\end{tikzpicture}
\caption{$(b_1,p_{i+1})$ blocks $(p_{i+2},p_i)$}
\label{fig:1D-alg-greedy-i-i+2-blocked}
\end{figure}

In the following, we determine the leftmost and the rightmost point of $C_G(p_i,p_{i+2})$ and their distances to $p_i$ and $p_{i+1}$.

Let $b_1,\ldots, b_j$ be the points with a back edge to $p_{i+1}$ ordered from left to right. 
So, $b_{j}$ is the rightmost point with a back edge to $p_{i+1}$. 
Due to maximality this point exists, otherwise $b_j = p_{i+1}$.    

We call these points a ``sequence of blockers'', since for each $2\leq k \leq j$ the edge $(b_k,p_{i+1})$ blocks $(b_{k-1},p_{i})$ (see \Cref{fig:1D-alg-greedy-blockerchain,fig:1D-alg-greedy-distance-blockertuple}). Therefore, it holds 
\begin{equation*}
b_k - p_{i+1}\leq b_{k-1}- p_{i},
\end{equation*} 
which bounds their distance by
\begin{equation}\label{eq:1D-greedy-choice-distance-blockertuple}
b_k- b_{k-1}\leq p_{i+1}- p_{i}.
%  \tag{$\circ^\thelabeledequotationcounter$}\stepcounter{labeledequotationcounter}
\end{equation}    

\begin{figure}[ht]
\centering
\begin{tikzpicture}[graph,directed, node distance=4em, every node/.append style={minimum width=3.5em}]
	\node[] (i+1) {$p_{i+1}$};
	\node[right= of i+1] (i+2) {$p_{i+2}$};
	\node[right= of i+2] (k) {$b_1$};
	\node[right= of k] (k+1) {$b_{j-1}$};
	\node[right= of k+1] (kj) {$b_{j}$};
	% \node[right= of kj+1] (kj+2) {$p_{k_{j+2}}$};
	
	\path[]
	(i+1) edge (i+2)
	(i+2) edge [dotted]  (k)
	(k) edge  [dotted]  (k+1)
	(k+1) edge [dotted] (kj)
	(k) edge[bend right] (i+1)
	(k+1) edge[bend right] (i+1)
	(kj) edge[bend right] (i+1)
	(kj) edge[bend right] (k+1)
	;
\end{tikzpicture}
\caption{Let $b_1,\ldots, b_j$ be the points with a back edge to $p_{i+1}$.}% For each $2\leq k \leq j$, it holds that $(b_k,p_{i+1})$ blocks $(b_{k-1},p_{i})$.}
\label{fig:1D-alg-greedy-blockerchain}
\end{figure}

\begin{figure}[ht]
\centering
\begin{tikzpicture}[graph,directed, node distance=4em, every node/.append style={minimum width=3.5em}]
\node[] (i) {$p_{i}$};
\node[right= of i] (i+1) {$p_{i+1}$};
\node[right=6em of i+1] (kj-1) {$b_{k-1}$};
\node[right=5em of kj-1] (kj) {$b_{k}$};

\path[]
(kj-1) edge[bend right, firstColor, dashed] (i)
(i) edge (i+1)
(i+1) edge[dotted] (kj-1)
(kj-1) edge[bend right] (i+1)
(kj) edge[bend right,] (i+1)     
(kj-1) edge [dotted] (kj)
;

\path[]
($(kj-1.south)-(0,.5)$) edge[-|] node[below] {$\overset{\text{eq.~\ref{eq:1D-greedy-choice-distance-blockertuple}}}{\geq } b_{k}-p_{i+1}$} ($(i.south)-(0,.5)$)
($(kj.south)-(0,.5)$) edge[|-|]  node[below] {$\overset{\text{eq.~\ref{eq:1D-greedy-choice-distance-blockertuple}}}{\leq} p_{i+1}-p_i$} ($(kj-1.south)-(0,.5)$)
;
\end{tikzpicture}
\caption{$(b_{k},p_{i+1})$ blocks $(b_{k-1},p_{i})$ for  $2\leq k\leq j$}
\label{fig:1D-alg-greedy-distance-blockertuple}
\end{figure}

Further,  for $2\leq k\leq j-1$, the greedy algorithm added $(b_k,p_{i+1})$ to $E$ and discarded $(b_{k+1},b_{k-1})$ (compare to \Cref{fig:1D-alg-greedy-distance-blockertriple}). Therefore, their distance is bounded by
\begin{equation}\label{eq:1D-greedy-choice-distance-blockertriple}
b_k- p_{i+1} \leq b_{k+1} - b_{k-1}.
\end{equation}

\begin{figure}[ht]
\centering
\begin{tikzpicture}[graph,directed, node distance=4em, every node/.append style={minimum width=3.5em}]
\node[] at (0,0) (i+1) {$p_{i+1}$};
\node at (2,0) (kj-1) {$b_{k-1}$};
\node at (4,0) (kj) {$b_k$};
\node at (6,0) (kj+1) {$b_{k+1}$};

\path[]
(kj+1) edge[bend right, firstColor, dashed] (kj-1)
(i+1) edge[dotted] (kj-1)
(kj-1) edge[dotted] (kj)
(kj) edge[dotted] (kj+1)
(kj-1) edge[bend right] (i+1)
(kj) edge[bend right] (i+1)     
(kj+1) edge[bend right] (i+1)      
;

\path[]
($(i+1.south)-(0,.5)$) edge[|-|] node[below] {$\overset{\text{eq.~\ref{eq:1D-greedy-choice-distance-blockertriple}}}{\leq} b_{k+1}-p_{i+1}$} ($(kj.south)-(0,.5)$)
;
\end{tikzpicture}
\caption{Since $(b_k,p_{i+1})\in E$, then  $(b_{k+1},b_{k-1})\notin E$ for $2\leq k\leq j-1$}
\label{fig:1D-alg-greedy-distance-blockertriple}
\end{figure}

Now, we use these equations to bound the length of the shortest closed walk in $G$ containing $p_i$ and $p_{i+2}$. 
Let $p_l$ be the leftmost and $p_r$ the rightmost point  of $C_G(p_i,p_{i+2})$.
We do a case distinction on \ref{case:1D-alg-greedy-1-arc} $l= i$, and \ref{case:1D-alg-greedy-2-arc} $l<i$.

\begin{enumerate}[i)]
\item\label{case:1D-alg-greedy-1-arc} Due to maximality, for $l= i$, it holds that $p_r=b_{j}$ (see \Cref{fig:1D-alg-greedy-1-arc}). 

Therefore, we have
\[ |C_G(p_i,p_{i+2})| =  (b_{j}-p_i)\cdot 2 = (b_{j} - b_{j-1} + b_{j-1} - p_{i+1} +p_{i+1} -  p_i) \cdot 2.\]

If $j-2\geq 1$, using \Cref{eq:1D-greedy-choice-distance-blockertriple}, gives us
\[  b_{j-1} - p_{i+1} \leq b_{j} - b_{j-2}  \overset{\text{eq.~\ref{eq:1D-greedy-choice-distance-blockertuple}}}{\leq}  2 (p_{i+1} -  p_i) \leq 2 (p_{i+2} -  p_i) .\]

Otherwise, it holds $j<3$, then we have 
\[ |C_G(p_i,p_{i+2})| =  (b_{j}-p_i)\cdot 2 \overset{\text{eq.~\ref{eq:greedy-choice-i+1-k}, \ref{eq:1D-greedy-choice-distance-blockertuple}}}{\leq} 2 (p_{i+2} -  p_i)  .\]

This provides the dilation 

\begin{align*}
\odil(p_i,p_{i+2})&=\frac{|C_G(p_i,p_{i+2})|}{|\Delta(p_i,p_{i+2})|}  =\frac{(b_{j}-p_i)\cdot 2}{(p_{i+2}-p_i)\cdot 2}\\
& \leq    \frac{b_{j} - b_{j-1} +2 (p_{i+2} -  p_i) +p_{i+1} -  p_i}{p_{i+2}-p_i}\\
& \overset{\text{eq.~\ref{eq:1D-greedy-choice-distance-blockertuple}}}{\leq} \frac{2 (p_{i+2} -  p_i) +2(p_{i+1} -  p_i)}{p_{i+2}-p_i}  \leq 4
\end{align*}     
\begin{figure}[ht]
\centering
\begin{tikzpicture}[graph,directed, node distance=4em, every node/.append style={minimum width=3.5em}]
	\node[] (i) {$p_{i}$};
	\node[right= of i] (i+1) {$p_{i+1}$};
	\node[right= of i+1] (kj-1) {$b_{j-2}$};
	\node[right= of kj-1] (kj) {$b_{j-1}$};
	\node[right= of kj] (kj+1) {$b_{j}$};
	
	\path[]
	(i) edge (i+1)
	(i+1) edge [dotted] (kj-1)
	(kj-1) edge[dotted] (kj)
	(kj) edge[dotted]  (kj+1)
	(kj-1) edge[bend right] (i+1)
	(kj) edge[bend right] (i+1)
	(kj+1) edge[bend right] (i+1)
	(kj+1) edge[bend right] (i)
	;
	
	\path[node distance=5em]
	($(i+1.south)-(0,1em)$) edge[|-|]  node[below] {$\overset{\text{eq.~\ref{eq:1D-greedy-choice-distance-blockertriple}}}{\leq} b_j - b_{j-2}\leq 2 (p_{i+1} -  p_i) $}  ($(kj.south)-(0,1em)$)
	($(kj.south)-(0,1em)$) edge[-|]  node[below] {$\overset{\text{eq.~\ref{eq:1D-greedy-choice-distance-blockertuple}}}{\leq} p_{i+1} - p_i$}  ($(kj+1.south)-(0,1em)$)       
	;
\end{tikzpicture}
\caption{If \ref{case:1D-alg-greedy-1-arc} $l=i$, then $p_r=b_{j}$ and $(b_{j},p_{i})\in C_G(p_i,p_{i+2})$}
\label{fig:1D-alg-greedy-1-arc}
\end{figure}

\item\label{case:1D-alg-greedy-2-arc} If   $l<i$  and $G$ is maximal, then $ (p_{i+1},p_l)\in C_G(p_i,p_{i+2})$. In this case, since $b_1$ is the leftmost point with a back edge to $p_{i+1}$, it holds that $p_r=b_1$ (see \Cref{fig:1D-alg-greedy-2-arc}).    
By definition, $(p_{i+1},p_l)$ is the blocker of $(b_{j},p_i)$. Therefore,
\begin{equation}\label{eq:greedy-choice-i+1-l}
p_{i+1} - p_l \leq b_{j} - p_i \leq 4(p_{i+2}-p_i),
\end{equation}
where the last step follows from case~\ref{case:1D-alg-greedy-1-arc}. Therefore, the dilation is
\begin{align*}
\odil(p_i,p_{i+2})=&\frac{|C_G(p_i,p_{i+2})|}{|\Delta(p_i,p_{i+2})|}=\frac{(b_1 - p_l)\cdot 2}{(p_{i+2}-p_i)\cdot 2} = \frac{b_1 - p_{i+1} +p_{i+1} -  p_l}{p_{i+2}-p_i} \\   
&\overset{\text{eq.~\ref{eq:greedy-choice-i+1-k},\ref{eq:greedy-choice-i+1-l}}}{\leq}   \frac{p_{i+1}-p_i + 4(p_{i+2}-p_i)}{p_{i+2}-p_i} \leq 5. %\tag*{\qed}
\end{align*}    
\begin{figure}[ht]
\centering
\begin{tikzpicture}[graph,directed, node distance=4em, every node/.append style={minimum width=3.5em}]
	\node[] (i) {$p_{i}$};
	\node[left= of i,xshift=-1em] (l) {$p_{l}$};
	\node[right= of i] (i+1) {$p_{i+1}$};
	%  \node[right= of i+1] (i+2) {$p_{i+2}$};
	\node[right=of i+1] (k) {$b_1$};
	\node[right= of k] (kj+c) {$b_{j}$};
	
	\path[]
	(kj+c) edge[bend right, firstColor, dashed] (i)
	;
	
	\path[]
	(l) edge[dotted] (i)
	(i) edge (i+1)
	%   (i+1) edge (i+2)
	(i+1) edge[dotted] (k)
	%   (k) edge[bend right] (i+2)
	%  (k) edge[bend right] (i+1)
	(k) edge[dotted] (kj+c)
	(kj+c) edge[bend right] (i+1) 
	;
	
	\path[every edge/.append style={bend right}]
	(k) edge (i+1)
	(i+1) edge [] (l)
	(k) edge [] (i+1)
	;
	
	\path[]
	($(k.south)-(0,1.5em)$) edge[|-|]  node[below] {$\overset{\text{eq.~\ref{eq:greedy-choice-i+1-k}}}{\leq } p_{i+1}-p_{i}$} ($(i+1.south)-(0,1.5em)$)
	($(l.south)-(0,1.5em)$) edge[|-]  node[below] {$\overset{\text{eq.~\ref{eq:greedy-choice-i+1-l}}}{\leq } b_{j}-p_{i} \leq 4(p_{i+2}-p_i)$} ($(i+1.south)-(0,1.5em)$)
	;
	
\end{tikzpicture}
\caption{If \ref{case:1D-alg-greedy-2-arc}~$l<i$, then $p_r=b_1$ and $(b_1,p_{i+1}),(p_{i+1},p_l)\in C_G(p_i,p_{i+2})$. Note that $(p_{i+1},p_l)$ blocks $(b_{j+c},p_i)$.}
\label{fig:1D-alg-greedy-2-arc}
\end{figure}
\end{enumerate}

By maintaining the points in a doubly linked list and the relevant distances in a priority queue the runtime of \Cref{alg:greedy-1page-1dim} is $\bigo (n \log n)$ where $n$ is the number of points.
%The runtime of \Cref{alg:greedy-1page-1dim} is  $\bigo(n \log n )$ time: The tree is initialised with $n$ elements. In each iteration of the while-loop $1$ element is deleted. The runtime of an iteration is dominated by the tree operations where each costs $\bigo(\log n )$ time.
\end{proof}

%We conclude that \Cref{alg:greedy-1page-1dim} returns a maximal $1$-PPB $5$-spanner.

\begin{figure}%[ht]
\centering
\begin{tikzpicture}[graph, directed]
\node[] (1) {$p_1$};
\node[right= of 1] (2) {$p_2$};
\node[right= of 2] (3) {$p_3$};
\node[right= of 3] (4) {$p_4$};
\node[right= of 4] (5) {$p_5$};
\node[right= of 5] (6) {$p_6$};
\node[right= of 6] (7) {$p_7$};

\draw (1) edge node[below] {$3-5\varepsilon$} (2);
\draw (2) edge node[below] {$1$} (3);
\draw (3) edge node[below] {$\varepsilon$} (4);
\draw (4) edge node[below] {$1-3\varepsilon$} (5);
\draw (5) edge node[below] {$1-\varepsilon$} (6);
\draw (6) edge node[below] {$1-\varepsilon$} (7);

\path[every edge/.append style={bend right}]
(3) edge (1)
(5) edge (3)
(6) edge (3)
(7) edge (3)
;
\end{tikzpicture}
\caption{Example of a greedy spanner with dilation $t=\odil(p_2,p_4)=\frac{5-7\varepsilon}{1+\varepsilon}<5$}
%\caption{Example of a greedy spanner, its dilation is $t=\odil(p_2,p_4)=\frac{5-7\varepsilon}{1+\varepsilon}<5$.}
\label{fig:1PPB-greedy-worstcase-7}
\end{figure}

\Cref{fig:1PPB-greedy-worstcase-7} shows a point set $P$ and its greedily constructed spanner $G$.  The oriented dilation of $G$ is $t= \frac{5-7\varepsilon}{1+\varepsilon}$. 
%Since $\varepsilon$ can be arbitrary small, 
Thus, $G$ is a $5$-spanner for $P$. %Therefore, the bound of $t< 5$ is tight due to this example point set.

\begin{figure}%[ht]
\centering
\begin{tikzpicture}[graph, directed]
\node[] (1) {$p_1$};
\node[right= of 1] (2) {$p_2$};
\node[right= of 2] (3) {$p_3$};
\node[right= of 3] (4) {$p_4$};
\node[right= of 4] (5) {$p_5$};
\node[right= of 5] (6) {$p_6$};
\node[right= of 6] (7) {$p_7$};

\draw (1) edge node[below] {$3-5\varepsilon$} (2);
\draw (2) edge node[below] {$1$} (3);
\draw (3) edge node[below] {$\varepsilon$} (4);
\draw (4) edge node[below] {$1-3\varepsilon$} (5);
\draw (5) edge node[below] {$1-\varepsilon$} (6);
\draw (6) edge node[below] {$1-\varepsilon$} (7);

\path[every edge/.append style={bend right}]
(5) edge (1)
(5) edge (3)
(5) edge (2)
(7) edge (5)
;
\end{tikzpicture}
\caption{A spanner for the point set of \Cref{fig:1PPB-greedy-worstcase-7} with dilation $t=\odil(p_2,p_4)=\frac{2-2\varepsilon}{1+\varepsilon}<2$.}
\label{fig:1PPB-greedy-worstcase-7-nonopt}
\end{figure}

However, \Cref{fig:1PPB-greedy-worstcase-7-nonopt} shows a $t$-spanner with $t=\frac{2-2\varepsilon}{1+\varepsilon}<2$ for the same point set.  
Therefore, \Cref{alg:greedy-1page-1dim} does not return the minimum spanner for every given one-dimensional point set.\\

%%%%%%%%%%%%%%%%%%%%%%%%%%%%%%%

\subsubsection{Minimum One-Page Plane  Spanner}

Next, we focus on minimum dilation one-page plane spanners, which, unlike non-plane oriented spanners for one-dimensional point sets, do not have a minimum dilation of $1$ (\Cref{theo:bound-no-1-1d-1-2-page-plane}, contrasting with \Cref{theo:1D-1-spanner}).

%Since a $1$-PPB spanner has a smaller dilation than a spanner with the same edge set but another orientation (\Cref{theo:baseline}) and maximality decreases dilation, 
Due to \Cref{theo:baseline}, we restrict our attention to $1$-PPB spanners. The following dynamic program calculates a maximal $1$-PPB spanner with minimal dilation. 

\subparagraph{Idea.} 
%Due to one-page-planarity, if $(p_r,p_l)\in E$, it holds  that $(p_j,p_i)\notin E$ for $l<i<r<j$ and $i<l<j<r$.
Due to one-page-planarity, a point set $\{p_l,\ldots,p_r\}$ can be separated into two (almost) independently subproblems, $\{p_l,\ldots,p_k\}$ and $\{p_k,\ldots,p_r\}$, by adding the edges $(p_r,p_k)$ and $(p_k,p_l)$ to a split point $p_k$ between $p_l$ and $p_r$ (consider \Cref{fig:opt-alg-split-mid}).  We test all candidates for $p_k$ to minimise the dilation $t=\max\{t',t'',\odil(p_{k-1},p_{k+1})\}$ (\Cref{theo:odil-1D-1sided})  for $\{p_l,\ldots,p_r\}$ where $t'$ is the minimum dilation for $\{p_l,\ldots,p_k\}$ and $t''$ is the minimum dilation for $\{p_k,\ldots,p_r\}$.   
To compute $ \odil(p_{k-1},p_{k+1})$, we need to know the length of $C_G(p_{k-1},p_{k+1})$ which is, since $(p_r,p_k),(p_k,p_l) \in E$,  the union of $C_G(p_k,p_{k+1})$ and $C_G(p_{k-1},p_k)$. By the construction holds $|C_G(p_{k-1},p_k)|=2(p_k - p_{k_r})$, where $p_{k_r}$ is the rightmost point with $(p_k,p_{k_r})\in E$, and analogously, we have $|C_G(p_k,p_{k+1})|=2(p_{k_l}-p_k)$, where $p_{k_l}$ is the leftmost point with $(p_{k_l},p_k)\in E$.  
% It holds
%   \[ \odil(p_{k-1},p_{k+1})= \frac{|C_G(p_k,p_{k+1})|+|C_G(p_{k-1},p_k)|}{2(p_{k+1}-p_{k-1})}= \frac{2(p_{k_l}-p_{k_r)}}{2(p_{k+1}-p_{k-1})}.\] 
To include those two edges in the dynamic program, we consider the subproblem of computing $\odil(l,l',r',r)$, which denotes the oriented dilation of the subgraph on $\{p_l,\ldots,p_r\}$  under the assumption that $(p_{l'},p_l)\in C_G(p_l,p_{l+1})$ and $(p_{r},p_{r'})\in C_G(p_{r-1},p_{r})$.

\subparagraph{Recursion.} 
We add $(p_r,p_l)$ and distinct cases by the indices $l$,$l'$,$r'$ and $r$.
\begin{enumerate}[a)]
\item For $l'<r$ and $r'>l$, we test every point $p_k$ with  $l+2\leq k\leq r-2$ and do recursively on $\{p_l,\ldots,p_k\}$ and $\{p_k,\ldots,p_r\}$ (compare to \Cref{fig:opt-alg-split-mid}).
Therefore, we test every point $p_{k_r}$ with $l\leq k_r\leq k-2$ as leftmost point of $C_G(p_{k-1},p_{k})$ and every point $p_{k_l}$ with $k+2\leq k_l\leq r $ as rightmost point of $C_G(p_{k},p_{k+1})$. We choose the combination $k_r,k,k_l$ that minimises the dilation 
\[\odil(l,l',r',r)=%\underset{l+2\leq k\leq r-2}{\argmin} \underset{l\leq k_r\leq k-2}{\argmin} \underset{k+2\leq k_l\leq r}{\argmin}
\underset{k_r,k,k_l}{\min}\max\{\odil(l,l',k_r,k), \odil(k,k_l,r',r),\frac{p_{k_l}-p_{k_r}}{p_{k+1}-p_{k-1}}\}. \]

\begin{figure}
\centering
\begin{subfigure}{\linewidth}
\begin{tikzpicture}[graph, every node/.append style={minimum size=3em}, directed]
	\node[] (l) {$p_{l}$};
	\node[right= of l] (l') {$p_{l'}$};
	\node[right=2em of l'] (kr) {$p_{k_r}$};
	\node[right= of kr] (k) {$p_{k}$};
	\node[right= of k] (kl) {$p_{k_l}$};
	\node[right=2em of kl] (r') {$p_{r'}$};
	\node[right= of r'] (r) {$p_{r}$};
	
	\path[]
	(l) edge[dotted] (l')
	(l') edge[dotted] (kr)
	(kr) edge[dotted] (k)
	(k) edge[dotted] (kl)
	(kl) edge[dotted] (r')
	(r') edge[dotted] (r)           
	; 
	; 
	
	\draw (l') edge[bend right] (l);
	\draw (kl) edge[bend right] (k);
	\draw (k) edge[bend right] (kr);
	\draw (r) edge[bend right] (r');
	\draw (r) edge[bend right] (k);
	\draw (k) edge[bend right] (l);
	\draw (r) edge[bend right] (l);
\end{tikzpicture}
\caption{Split at $p_k$, recursion on both sides}%, $\odil(l,l',r',r)=\underset{k_r,k,k_l}{\min}\max\{\odil(l,l',k_r,k), \odil(k,k_l,r',r),\odil(p_{k-1},p_{k+1})\}$}%If $l'<r$ and $r'>l$, 
% \caption{optimal edge (oE) set for $(l,l',r',r)$ $=$ $\{\textsc{oE}(l,l',k_r,k) \cup \textsc{oE}(k,k_l,r',r) \cup  \{(p_r,p_l)\}$ }
\label{fig:opt-alg-split-mid}
\end{subfigure}
\begin{subfigure}{\linewidth}
\centering
\begin{tikzpicture}[graph, every node/.append style={minimum size=3em}, directed]
\node[] (l) {$p_{l=r'}$};
\node[right= of l] (l') {$p_{l'}$};
\node[right= of l'] (kr) {$p_{k_r}$};
\node[right= of kr] (r-1) {$p_{r-1}$};
\node[right= of r-1] (r) {$p_{r}$};

\path[]
(l) edge[dotted] (l')
(l') edge[dotted] (kr)
(kr) edge[dotted] (r-1)
(r-1) edge (r)
; 

\draw (l') edge[bend right] (l);
\draw (r-1) edge[bend right] (kr);
\draw (r-1) edge[bend right] (l);
\draw (r) edge[bend right] (l);
\end{tikzpicture}
\caption{Split at $p_{r-1}$, just left recursion}%, $\odil(l,l',r',r) =\underset{l\leq k_r\leq r-3}{\min}\odil(l,l',k_r,r-1)$}
%     \caption{$\textsc{oE}(l,l',l,r) = \textsc{oE}(l,l',k_r,r-1) \cup \{(p_r,p_l)\}$}%If $l'<r$ and $r'=l$, 
\label{fig:opt-alg-split-right}
\end{subfigure}
\begin{subfigure}{\linewidth}
\centering
\begin{tikzpicture}[graph, every node/.append style={minimum size=3em}, directed]
\node[] (l) {$p_{l}$};
\node[right= of l ] (l+1) {$p_{l+1}$};
\node[right= of l+1] (kl) {$p_{k_l}$};
\node[right= of kl] (r') {$p_{r'}$};
\node[right= of r'] (r) {$p_{r=l'}$};

\path[]
(l) edge (l+1)
(l+1) edge[dotted] (kl)
(kl) edge[dotted] (r')
(r') edge[dotted] (r)           
; 

\draw (r) edge[bend right] (r');
\draw (kl) edge[bend right] (l+1);
\draw (r) edge[bend right] (l+1);
\draw (r) edge[bend right] (l);
\end{tikzpicture}
\caption{Split at $p_{l+1}$, just right recursion}%, $\odil(l,l',r',r) = \underset{l+2\leq k_l\leq r}{\min} \odil(l+1,k_l,r',r)$}
%  \caption{$\textsc{oE}(l,r,r',r) = \textsc{oE}(l+1,k_l,r',r) \cup \{(p_r,p_l)\}$}%If $l'=r$ and $r'>l$, 
\label{fig:opt-alg-split-left}
\end{subfigure}
\caption{Visualisation of the recursion of the dynamic program}
\end{figure}

\item For $r'=l$ (and $l'<r$), maximality implies that $(p_{r-1},p_{l})$ is in the back edge set (compare to \Cref{fig:opt-alg-split-right}). Therefore, we consider recursively $\{p_l,\ldots, p_{r-1}\}$. We choose the point $p_{k_r}$ as leftmost point for $C_G(p_{r-2},p_{r-1})$ which minimises the dilation  \[\odil(l,l',r',r)=\underset{l\leq k_r\leq r-3}{\min} \max\{ \odil(l,l',k_r,r-1), \frac{p_{r}-p_{l}}{p_{r}-p_{r-2}}\} .\]

\item Analogous, for  $l'=r$ (and $r'>l$), we consider recursively $\{p_{l+1},\ldots, p_r\}$ (compare to \Cref{fig:opt-alg-split-left}). We test every point $p_{k_l}$ as rightmost point for $C_G(p_{l+1},p_{l+2})$  to minimise the dilation  \[\odil(l,l',r',r)=\underset{l+3\leq k_l\leq r}{\min} \max\{ \odil(l+1,k_l,r',r), \frac{p_{r}-p_{l}}{p_{l+2}-p_{l}}\} .\]
\end{enumerate}

\subparagraph{Termination.} 
Each recursion terminates with a forbidden subproblem $(l,l',r',r)$ or with a set of three points.
Invalid combinations are
\begin{enumerate}[i)]
\item  $l'<l+2$, $r'>r-2$ or $l>r$ which contradicts the definition of the indices,
\item $r< l+2$ which implies $|P|< 3$, or
\item $l'<r$, $r'>l$ and $l'>r'$ which contradicts planarity.
\end{enumerate}
If $l'=r$ and $r'=l$, due to maximality, it holds that $|P|=3$. The only (and optimal) solution for a set of three points is the triangle $\Delta p_lp_{l+1}p_{r=l+2}$.

\subparagraph{Runtime and Space.} 
The dynamic program goes through all $\bigo(n^7)$ combinations of $l,l',k_r,k,k_l,r'$ and $r$. By storing the dilation of each subproblem $(l,l',r',r)$, choosing the optimal combination of $k,k_r$ and $k_l$ does not take additional time. By back tracking, the edge set of a minimum $1$-PPB spanner is calculated in $\bigo(n^7)$ time and  $\bigo(n^4)$ space.\\

We conclude that this dynamic program returns an optimal $1$-PPB graph. 
\begin{theorem}[Optimal $1$-PPB]
\label{theo:opt-alg-1PPB}
Given a one-dimensional point set $P$ of size $n$, a minimum one-page plane oriented spanner for $P$ can be calculated in $\bigo(n^7)$ time.
\end{theorem}
%%%%%%%%%%%%%%%%%%%%%%%%%%%%%%%

\section{Two-Dimensional Point Sets}\label{sec:2D}

In the one-dimensional case, we have seen that a $1$-spanner exists for every point set, though it is not plane. In two dimensions, the complete bi-directed graph is always a directed $1$-spanner. However, this is not the case for oriented dilation. There are point sets for which no $1$-spanner exists, even more, no $1.465$-spanner exists.

\begin{observation}
\label{theo:2D-upperlowerbound}
%There are two-dimensional point sets for which no oriented $1$-spanner exists.
There are two-dimensional point sets for which no oriented $t$-spanner exists for $t<2\sqrt{3}-2 < 1.465$.
\end{observation}
\begin{proof}
We show that $t=2\sqrt{3}-2$  holds for every $t$-spanner for the point set $P$ where $p_1$, $p_2$ and $p_4$ form an equilateral triangle and $p_3$ is its centre. % $P=\{p_1=(-\frac{1}{2}, -\frac{\sqrt{3}}{4}),~~ p_2=(\frac{1}{2}, -\frac{\sqrt{3}}{4}),~~ p_3=(0,0),~~ p_4=(0, \frac{\sqrt{3}}{4})\}$     (a equilateral triangle with $(0,0)$ as centre).
% The optimal cycle $\Delta(p_i,p_j)$ with $i,j\in\{1,2,4\}$ is $\Delta p_ip_3p_j$ or $\Delta p_jp_3p_i$.
% Each of these triangles is optimal for $\Delta(p_3,p_i)$. If there is an oriented $1$-spanner $G=(P,E)$, these triangles have to be contained in $E$.
% Ignoring the orientation, the union of these triangles is the complete graph $K_4$. Therefore, $E$ contains no further edges. 
% It remains to orientate the edges. But, there is no possibility to orientate the edges of $K_4$ such that for each $p_i,p_{j}\in\{p_1,p_2,p_4\}$ the optimal triangle $\Delta p_ip_3p_j$ or $\Delta p_jp_3p_i$ is contained.
% \begin{figure}
%     \centering
%     \begin{tikzpicture}[graph]
%         \node[] (1) {$p_1$};
%         \node[right=7em of 1] (2) {$p_2$};
%         \node[above= of 1] at ($(1)!.5!(2)$) (3) {$p_3$};
%         \node[above= of 3] (4) {$p_4$};
%         \path[directed]
%         (2) edge (1)
%         (1) edge (3)
%         (3) edge (2)        
%         (3) edge (4)
%         (4) edge (1)
%         ;
%         \draw (4) edge[dashed] (2);
%     \end{tikzpicture}
%     \caption{No possibility to orientate the faces of $K_4$ consistently}
%     \label{fig:k4-orientation}
% \end{figure}
Since more edges can only improve the dilation, the minimum spanner must be a complete graph $K_4$. 
%  Since a $K_4$ contains exact two consistently oriented triangles
Taking into account mirroring and rotation, \Cref{fig:2D-upperlowerbound} lists all orientations for $K_4$. For each orientation, we list the oriented dilation of each pair of points in a table. We see that every graph is a $(2\sqrt{3}-2)$-spanner.
\end{proof}

\begin{figure}
\centering
\begin{minipage}{.45\linewidth}\centering
\includegraphics[page=4]{figures/equiliteral-triangle.pdf}
\\
\begin{tabular}{c | c |  c | c}
$\odil$   &  $p_1$ &  $p_2$ &  $p_3$  \\ \hline \hline
$p_4$ & $1$  &  $2\sqrt{3}-2$ & $1$ \\ \hline
$p_3$ & $1$ & $1$  &  \\\cline{1-3}
$p_2$ & $1$  & \multicolumn{1}{c}{}  & \\\cline{1-2} %\blankcols{1}
\end{tabular}
\end{minipage}
\begin{minipage}{.45\linewidth}\centering
\includegraphics[page=1]{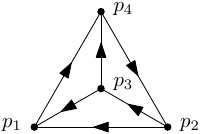}\\
\begin{tabular}{c | c |  c | c}
$\odil$   &  $p_1$ &  $p_2$ &  $p_3$  \\ \hline \hline
$p_4$ & $6\sqrt{3}-9$ &  $1$ & $1$ \\\hline
$p_3$ & $2\sqrt{3}-2$ & $1$  &  \\\cline{1-3}
$p_2$ &  $6\sqrt{3}-9$ & \multicolumn{1}{c}{}  & \\\cline{1-2} %\blankcols{1}
\end{tabular}
\end{minipage}\\
\begin{minipage}{.45\linewidth}\centering
\includegraphics[page=3]{figures/equiliteral-triangle.pdf}
\\
\begin{tabular}{c | c |  c | c}
$\odil$   &  $p_1$ &  $p_2$ &  $p_3$  \\ \hline \hline
$p_4$ & $1$  &  $2\sqrt{3}-2$ & $1$ \\ \hline
$p_3$ & $1$ & $1$  &  \\\cline{1-3}
$p_2$ & $1$  & \multicolumn{1}{c}{}  & \\\cline{1-2} %\blankcols{1}
\end{tabular}
\end{minipage}
\begin{minipage}{.45\linewidth}\centering
\includegraphics[page=2]{figures/equiliteral-triangle.pdf}\\
\begin{tabular}{c | c |  c | c}
$\odil$   &  $p_1$ &  $p_2$ &  $p_3$  \\ \hline \hline
$p_4$ & $6\sqrt{3}-9$ &  $1$ & $1$ \\\hline
$p_3$ & $2\sqrt{3}-2$ & $1$  &  \\\cline{1-3}
$p_2$ &  $6\sqrt{3}-9$ & \multicolumn{1}{c}{}  & \\\cline{1-2} %\blankcols{1}
\end{tabular}
\end{minipage}
\caption{All possible orientations for $K_4$}
\label{fig:2D-upperlowerbound}
\end{figure}

In this section, we consider \emph{consistent} orientations. This means each face of an oriented graph is confined by an oriented cycle. If such an orientation exists for a given graph, we call it an \emph{orientable graph}.

While an oriented complete graph does not lead to a $1$-spanner, orienting its triangles greedily consistently yields a $2$-spanner.
\begin{proposition}
\label{theo:2D-alg-complete-graph}
%$\Cref{alg:greedy-complete-graph} orients the complete graph $K_n$ such that $K_n$ is a $2$-spanner
For every point set an oriented $2$-spanner can be constructed in $\mathcal{O}(n^3\log n)$ time by orienting a complete graph.
\end{proposition}
\begin{proof}
We prove that \Cref{alg:greedy-complete-graph} computes such a spanner. Consider an arbitrary pair of points $p, p'$. Let $p''$ be a point that forms a triangle of smallest length with $p$ and $p'$. Consider the iteration of the for-loop processing $\Delta_{pp'p''}$. 

If at the beginning of the iteration at most one edge of $\Delta_{pp'p''}$ was oriented, then $\Delta_{pp'p''}$ is oriented consistently, and $\odil(p,p') = 1$. Otherwise, at least two of the edges of $\Delta_{pp'p''}$ were already oriented, and therefore are incident to previously processed, consistently oriented triangles. By concatenating the two triangles, we obtain a closed walk containing $p$ and $p''$, and since both triangles were processed first, their lengths is at most the length of $\Delta_{pp'p''}$, and the closed walk is at most twice as long. Thus, $\odil(p,p') \leq 2$. 

\Cref{alg:greedy-complete-graph} takes $\mathcal{O}(n^3\log n)$ time because of sorting $\binom{n}{3}$ triangles.
\end{proof}

\begin{algorithm}
\caption{Orientation of $K_n$}\label{alg:greedy-complete-graph}
\begin{algorithmic}%[1]
\REQUIRE two-dimensional set $P$ of $n$ points
\ENSURE orientation of the complete graph $K_n$ such that $K_n$ is a $2$-spanner
\STATE Sort the triangles in  $D= \{\Delta_{pp'p''}\mid \forall p,p',p''\in P\}$  ascending by lengths \COMMENT{all possible triangles in $K_n$}
\FORALL{$\Delta_{pp'p''}\in D$}
\IF{no edge of $\Delta_{pp'p''}$ is oriented}
\STATE orient $(p,p')$ arbitrarily
\STATE orient $\Delta_{pp'p''}$, the orientation is fixed by $(p,p')$
\ELSIF{$1$ edge of $\Delta_{pp'p''}$ is oriented}
\STATE orient $\Delta_{pp'p''}$, the orientation is fixed by this edge
\ELSIF{$\geq 2$ edges of $\Delta_{pp'p''}$ are oriented}
\STATE skip $\Delta_{pp'p''}$
\ENDIF
\ENDFOR
\STATE orient the remaining edges arbitrarily
\end{algorithmic}
\end{algorithm}

%%%%%%%%%%%%%%%%%%%%%%%%%%%%%%%

\subsection{Hardness}\label{sec:2D-hardness}

We now have shown that although a $1$-spanner does not exist for every two-dimensional point set, we get a $2$-spanner via the oriented complete graph. However, this leads to a graph with $\frac{1}{2} n \cdot (n-1)$ edges, so this graph is neither sparse nor plane.

We will see that computing a minimum oriented spanner with a linear number of edges  is NP-hard.

\begin{theorem}
\label{theo:min-dil-np-hard}
Given a two-dimensional set $P$ of $n$ points and the parameters $t$ and $m$, it is NP-hard to decide if there is an oriented $t$-spanner $G=(P,E)$ with $|E|\leq m$.
\end{theorem}
\begin{proof}
We reduce from the NP-complete Hamilton circuit problem for grid graphs~\cite{ItaiPS82.hamilton}. 
We start with a polynomial-size node-induced subgraph $G'=(P',E')$ of a grid whose cells have side length one. We choose $P=P'$, $m=n$, where $n=|P|$, and $t=n \cdot c$,
where $c$ is a rational number such that 
\[ 
\frac{1}{2 + \sqrt{2}} \leq  c  <  \frac{1}{2+\sqrt{2}} + \frac{0.1}{n} < \frac{1}{2+\sqrt{2}} + \frac{1}{n} \cdot \frac{\sqrt{2} - 1}{2+\sqrt{2}}.
\]
To choose $t$, it is sufficient to compute the first $\bigo(\log (n))$ digits of $c$, which is possible in polynomial time.

We will show that there is an oriented $t$-spanner $G=(P,E)$ with $|E|\leq m$ if and only if $G'$ has a Hamiltonian cycle.

First, we show that if $G'$ has a Hamiltonian cycle, then there is an oriented $t$-spanner $G=(P,E)$ with $|E|\leq m$.
Let $E_{HC}$ be the edges of a Hamiltonian cycle in $G'$. We set $G=(P,E_{HC})$. It remains to bound the dilation of $G$. Since the $n$ edges of $G$ form a simple cycle, a shortest closed walk of any point tuple is this cycle. Further, all edges have length $1$. As the points in $P$ are embedded on a grid whose cells have side length one, it holds $|\Delta(p,p')| \geq  2+\sqrt{2}$ for all $p,p'\in P$.
Therefore, the dilation of $G$ is
\[ \underset{p,p'\in P}{\max}\frac{|C_G(p,p')|}{|\Delta(p,p')|} \leq \frac{n\cdot 1}{2+\sqrt{2}} \leq  t .\]

Now, we show that if $G'$ has no Hamiltonian cycle, then there is no oriented $t$-spanner $G=(P,E)$ with $|E|\leq m$.
Any strongly connected graph $G$ on $P$ with $m=n$ edges must be a cycle. 
Again, a shortest closed walk of any point tuple is this cycle.
So, the dilation of this graph $G=(P,E)$ is 
\[ \underset{p,p'\in P}{\max}\frac{|C_G(p,p')|}{|\Delta(p,p')|} =\frac{\underset{e\in E}{\sum} |e|}{ \underset{p,p'\in P}{\min} |\Delta(p,p')|}\]
W.l.o.g., we may assume that there is always an ``L-shaped'' point triple $p_1,p_2,p_3\in P$ with $|p_1-p_2|=|p_2-p_3|=1$ and $|p_1-p_3|=\sqrt{2}$. If no such triple exists, $G'$ is a one-dimensional path and the Hamilton circuit problem for this graph class is trivial. Therefore, we have
\[ \underset{p,p'\in P}{\min} |\Delta(p,p')| = 2+\sqrt{2}.\]
Since $G'$ has no Hamiltonian cycle, $G$ must contain an edge whose length is at least $\sqrt{2}$.
Therefore, the dilation of $G$ is at least
\begin{equation*}
\frac{\underset{e\in E}{\sum} |e|}{ 2+\sqrt{2}} \geq \frac{n-1+\sqrt{2}}{2+\sqrt{2}} > t. \tag*{\qedhere}
\end{equation*}
% Therefore, deciding whether there is an oriented $t$-spanner for $P$ with $m=n$ edges is  equivalent to deciding whether a node-induced grid graph has a Hamiltonian cycle.  
% We prove the problem is NP-hard even for $m=n$. We reduce the problem of constructing a minimum dilation oriented cycle ($m=n$) to the Euclidean Travelling Salesman Problem (\textsc{Euclidean-TSP} is NP-hard~\cite{Papadimitriou.1977}).
% The denominator in the definition of (oriented) dilation is fixed by the point set. So, the oriented dilation is minimised by minimising the largest shortest closed walk over all pairs of points. %the numerator
% Since a minimum dilation cycle is searched ($m=n$), the shortest closed walk is the same for each pair of points. Therefore, to compute a minimum oriented spanner with $n$ edges, is the same as searching for the shortest Hamilton cycle. That is \textsc{Euclidean-TSP}.
\end{proof}

For computing a minimum \emph{plane} oriented spanner, i.e.\ a minimum oriented dilation triangulation, hardness remains open. In the undirected setting this question is a long-standing open problem~\cite{brandt-mdt-2014,Eppstein00,Giannopoulos.2010}. To our knowledge, it is not even known whether a PTAS for this problem exists, and for several related open problems there is no FPTAS~\cite{Giannopoulos.2010}, unless $P=NP$. We show the following relative hardness result.

\begin{observation}
\label{observation:minimum_dilation_triangulation}
%Computing a minimum plane oriented  spanner is at least as hard as approximating minimum dilation triangulation.
An FPTAS for minimum plane oriented spanner would imply an FPTAS for minimum dilation triangulation.
\end{observation}

\begin{figure}[ht]
\centering
\begin{tikzpicture}[directed]
\node[dot] (l0) at (0, 0) {};
\node[dot]  (l1) at (1, 0) {};
\node[dot]  (l2) at (5, 0) {};
\node[dot]  (l3) at (9, 0) {};
\node[dot]  (l4) at (10, 0) {};

\node[dot]  (u0) at (.5, 1) {};
\node[dot]  (u1) at (4.5, 1) {};
\node[dot]  (u2) at (5.5, 1) {};
\node[dot]  (u3) at (9.5, 1) {};

\foreach \i [evaluate=\i as \j using {int(\i+1)}] in {0, 1, 2}
\path (u\i) edge (u\j);
\foreach \i [evaluate=\i as \j using {int(\i-1)}] in {1, 2, 3,4}{
\path (l\i) edge (l\j);
}
\path
(l0) edge (u0)
(u0) edge (l1)
(l2) edge (u1)
(u2) edge (l2)
(l3) edge (u3)
(u3) edge (l4)
;

\draw ($(l1.south)-(0,.3)$) edge[|-|]  node[below] {$c_2\cdot\varepsilon/n$} ($(l2.south)-(0,.3)$);
\draw ($(u0.north)+(0,.3)$) edge[|-|]  node[above] {$c_2\cdot\varepsilon/n$} ($(u1.north)+(0,.3)$);
\draw ($(l1.south)-(0,.3)$) edge[-|]  node[below] {$c_1\cdot\varepsilon^2/n$} ($(l0.south)-(0,.3)$);
\draw ($(l0.west)+(-.3,.3)$) edge[|-|]  node[above, sloped] {$c_1\cdot\varepsilon^2/n$} ($(u0.west)+(-.3,.3)$);
\end{tikzpicture}
\caption{By replacing each point with this construction, we can reduce the minimum dilation triangulation problem to the minimum plane oriented spanner problem.}
\label{fig:min-dil-tri-2d}
\end{figure}

We sketch a reduction from the minimum dilation triangulation problem.
Suppose we are given a set $P$ of $n$ points.
The idea is to replace every point with the gadget depicted in
\Cref{fig:min-dil-tri-2d}.
The gadget consists of $3\cdot\lceil\frac{4n}{3}\rceil$ points, positioned along a line in small triangles. Since the triangles are small relative to the distance between the triangles, we obtain an oriented dilation of less than $(1+c\varepsilon)$ within the gadget by connecting the gadget points as in the figure, where $c>0$ is a constant that we can choose by suitably picking the constants $c_1$ and $c_2$.
%
%Now by scaling the construction by a linear factor in $\varepsilon$/$n$ all distances in the construction are within $\varepsilon$.
%In order for the graph to have minimum oriented dilation, we need to connect the points as in \Cref{fig:min-dil-tri-2d}.

The gadget leaves us with room for two edges to every one of $n$ points, one in each direction, without disturbing planarity.
Therefore, minimising the oriented dilation in this setting while requiring planarity should pick the edges that correspond to those in the minimum dilation triangulation, except if they are $\varepsilon$-close.

%%%%%%%%%%%%%%%%%%%%%%%%%%%%%%%

\subsection{Greedy Triangulation}\label{sec:2D-greedy}

Given the relative hardness of computing minimum plane oriented spanners, we next inves\-tigate the question of whether plane oriented $\bigo(1)$-spanners exist. In the undirected setting,  prominent examples for plane constant dilation spanners are the Delaunay triangulation and the greedy triangulation. Our main result on oriented spanners in two dimensions is that the greedy triangulation is an $\bigo(1)$-spanner for convex point sets (i.e.\ point sets for which the points lay in convex position).

The \emph{greedy triangulation} of a point set is the triangulation obtained by considering all pairs of points by increasing distance, and by adding the straight-line edge if it does not intersect any of the edges already added. The greedy triangulation can be computed in linear time from the Delaunay triangulation~\cite{LEVCOPOULOS1999197}, and in linear time for a convex point set if the order of the points along the convex hull is given~\cite{LevcopoulosL92}. 

It is known that the dilation of the undirected greedy triangulation is bounded by a constant. This follows from the fact that the greedy triangulation fulfils the $\alpha$-diamond-property~\cite{DasJ89}. The currently best upper bound on the dilation $t$ of such a triangulation is $ t\leq \frac{8(\pi - \alpha)^2}{\alpha ^2 \sin^2(\sfrac{\alpha}{4})}$ with a lower bound on $\alpha$ of $\sfrac{\pi}{6}$ for the greedy triangulation~\cite{BoseLS07}.

\begin{figure}[ht]
\centering
\begin{tikzpicture}[every edge/.style={draw, semithick}, scale=.8]
\foreach \a in {2,3,...,6}{
\pgfmathtruncatemacro\b{\a+1}
\pgfmathtruncatemacro\c{\a+98}
\draw (\c*-360/60:-15) node[dot]  (v\a){};% [label={-175:{$p_{\b}$}}]
}

\node[dot] (v1) at ($(v2|-v6)+(1.35,0)$) {};%[label={-5:{$p_2$}}]
\node[dot] (v0) at ($(v1)+(1.5,-2.8)$) {};%[label={-90s:{$p_1$}}]

\node[graph/plain]  at ($(v2.east)+(.4,0)$) {$p_3$};
\node[graph/plain]  at ($(v6.south)+(0,-.4)$) {$p_n$};
\node[graph/plain]  at ($(v1.east)+(.4,0)$) {$p_2$};
\node[graph/plain]  at ($(v0.south)+(0,-.4)$) {$p_1$};

\draw  (v1) edge (v2);
\draw  (v1) edge (v3);
\draw  (v1) edge (v4);
\draw  (v1) edge (v5);
\draw  (v1) edge (v6);
\draw  (v6) edge (v5);
\draw  (v5) edge (v4);
\draw  (v4) edge (v3);
\draw  (v3) edge (v2);
\draw  (v2) edge (v4);
\draw  (v4) edge (v6);
\draw  (v6) edge (v2);
\draw (v0) edge (v1);
\draw (v0) edge (v6);

\draw[dashed, opacity=.5] (v0) edge (v2);

\foreach \a in {2,3,...,6}{
\pgfmathtruncatemacro\b{\a-1}
\draw ($(v0.south -|v\b)-(0,.8)$) edge[-|]  node[below] {$1$} ($(v0.south -|v\a)-(0,.8)$);
}

\draw ($(v0.south)-(0,.8)$) edge[|-|]  node[below] {$1+\varepsilon$} ($(v0.south -|v1)-(0,.8)$);

\draw ($(v0.east)+(.3,0)$) edge[|-|]  node[right] {$\delta$} ($(v0.east |-v1)+(.3,0)$);
\draw ($(v0.east |-v1)+(.3,0)$) edge[-|]  node[right] {$\delta'$} ($(v0.east |-v2)+(.3,0)$);
\end{tikzpicture}
\caption{Greedy triangulation for a non-convex point set. It is is also a  minimum weight triangulation.}
\label{fig:greedy-non-convex-omega-n}
\end{figure}

We first observe that restricting to convex point sets is necessary to obtain constant dilation from orienting the greedy triangulation. For this, consider the greedy triangulation $T=(P,E)$ on the following non-convex point set $P = \{p_1, \ldots, p_n\}$. As shown in \Cref{fig:greedy-non-convex-omega-n}, we place the points $\{p_3,\ldots,p_n\}$ on a very flat parabola.
By arbitrarily decreasing the $y$-distances $\delta,\delta'>0$, we reduce the construction to an almost one-dimensional problem.
The $x$-distance of each consecutive point pair $p_i,p_{i+1}$ is $1$ for $2\leq i <n$. 
We place $p_1$ slightly right of $p_2$ such that the $x$-distance is $1+\varepsilon$ for a small $\varepsilon>0$. Therefore, the $x$-distance of $p_1$ to $p_i$ is $\varepsilon$-larger than $x$-distance of $p_2$ to $p_{i+1}$ for $3\leq i <n$. By this, the greedy triangulation added  $(p_2,p_{i+1})\in E$ and discarded $(p_1,p_i)\notin E$ for $3\leq i <n$.
Further, we place $p_1$ slightly below $p_2$ so that, due to planarity, $p_1$ is only adjacent to $p_2$ and $p_n$. Since $p_1$ has degree~$2$, for any orientation of $T$, every shortest closed walk containing $p_1$ contains the subpath $p_n,p_1,p_2$ or vice versa.
From this construction, it follows that $\odil(p_1,p_3)\geq \frac{n-1+\varepsilon}{2+\varepsilon}$, for any orientation of $T$. 
Therefore, every orientation of greedy triangulation has oriented dilation $\Omega(n)$. 

In contrast, for convex point sets, a consistent orientation of the greedy triangulation results in a $\bigo(1)$-spanner. Since the dual graph of a triangulation of a convex point set is a tree, it is orientable. If the neighbourhood of each point is sorted and accessible in constant time, this consistent orientation can be computed in $O(n)$ time.

Next, for a convex point set, we show that a consistent orientation of the greedy triangulation has constant dilation. The proof works similarly to the proof of correctness for the greedy one-page plane spanner (see \Cref{theo:alg-greedy}).

\begin{theorem}
\label{theo:alg-2D-greedy-convex}
By orienting the greedy triangulation of a convex two-dimensional point set~$P$ consistently, we get a plane oriented $\bigo(1)$-spanner for $P$.
\end{theorem}
\begin{proof}%\setcounter{equation}{0}
Let $T=(P,E)$ be the greedy triangulation of $P$ and $G=(P, \overrightarrow{E})$ its consistent orientation (which is unique up to reversing all edges). Note that $T$ is an undirected graph, whereas $G$ is directed. To improve readability, undirected edges are written with curly brackets and directed edges with round brackets. 
Due to the $\alpha$-diamond property, the undirected dilation of any greedy triangulation can be bounded by a constant~\cite{DasJ89}. Let $t_g$ be the (smallest such) constant.
%By consistently orientation, each pair of points can reach each other. 

Let $p,p'$ be two points in $P$. We will prove $\odil(p,p')\in \bigo(1)$. 
We distinguish between whether i) $\{ p,p'\}$ is in ${E}$ or ii) not.

% \begin{enumerate}[label=\roman*)]
%     \item \label{case:alg-2D-greedy-p-p1-inside} 
%%Highlevel of proof%%%%
For i), we prove that there is a consistently oriented triangle, i.e.\ a closed walk containing $p$ and $p'$ of length in $\bigo (|\Delta(p,p')|)$, where $\Delta(p,p')$ is a smallest triangle in the complete graph incident to $p$ and $p'$. 
In other words, we bound the length of the shortest closed walk $C_G(p,p')$ in $G$. 

%%% proof %%%%        
Let $q\in P$ be the third point incident to $\Delta(p,p')$, i.e.\ $\Delta(p,p')= \Delta_{pp'q}$. 
Let $\Pi$ be the undirected path from $p$ to $q$ in $T$. Let $q'$ be the point, such that $\{p,q'\}$ is the first edge in $\Pi$, w.l.o.g.\ $q' \neq p'$.
(If $q' = p'$, then the proof would be the same with switched roles for $p$ and $p'$.) Due to the $\alpha$-diamond property, it holds that $|\Pi|\leq t_g |p-q|$ and therefore also $|p-q'|\leq t_g |p-q|$.

If $\{q',p'\}\in E$, then $\Delta_{pq'p'}\in E$. Since $\Delta_{pq'p'}$ is oriented consistently in $G$, we can bound 
%     $|C_G(p,p')| \leq |\Delta_{pq'p'}| \leq 2\cdot t_g \cdot |\Delta(p,p')|$.        
\begin{align*}
\odil(p,p')\leq &%\frac{|C_G(p,p')|}{|\Delta(p,p')|} \leq 
\frac{|p-p'|+|p'-q'|+|q'-p|}{|\Delta(p,p')|}  \overset{\text{triangle inequality}}{\leq } \frac{2|p'-p| + 2|p-q'|}{| \Delta_{pp'q}|}\\
\overset{\text{$\alpha$-diamond}}{\leq} &\frac{2 |p'-p| + 2 \cdot t_g |p-q|}{| \Delta_{pp'q}|}  ~{\leq} ~  \frac{2 \cdot t_g | \Delta_{pp'q}|}{| \Delta_{pp'q}|}~{\leq} ~ 2 \cdot t_g \quad \text{~for $\{q',p'\}\in E$.}
\end{align*}

However,  $\{q',p'\}$ could be blocked by another edge. 
Since $P$ is convex and $T$ is planar, this edge is incident to $p$. Let  $\{p,b_1\}$ be the (shortest) edge that blocks $\{q',p'\}$.
Analogous to \Cref{theo:alg-greedy}, we show that there could be a ``sequence of blockers'' (see \Cref{def:greedy-blocker}).  However, we will show that there is a consistent oriented triangle containing $p$ and $p'$ of length in $\bigo (|\Delta(p,p')|)$.  
Let $\Delta_{pp'b_{j+2c}}\in E$ be this triangle, i.e. $\{p,b_{j+2c}\}$ is the last edge in this ``sequence of blockers''.
Let the points $b_1,\ldots, b_{j+2c}$ be ordered such that $\{p,b_i\}\in E$ is the shortest edge that blocks $\{p',b_{i-1}\}$ with $b_0=q'$ (see \Cref{fig:2D-greedy-orientation-Cg}).

\begin{figure}[ht]
\centering
\begin{tikzpicture}[every edge/.style={graph edge, semithick}, scale=.8]
\node [dot,label={-90:{$p$}}] (p) at (-2.5,-2.5) {};
\node [dot,label={90:{$q'$}}] (p') at (-5,2) {};
\node [dot,label={-90:{$p'$}}] (q') at (2,-4) {};
\node [dot,label={3:{$b_1$}}] (b1) at (-3,2.5) {};

\node [dot,label={0:{$b_{j+2c}$}}] (bj+2c) at (5,-2.5) {};
\node [dot,label={0:{$b_2$}}] (b2) at (-1,2.5) {};
\node [dot,label={0:{$b_{j}$}}] (bj) at (3,1.5) {};
\node [dot,label={0:{$b_{j+1}$}}] (bj+1) at (4,0.5) {};

\draw (p') edge (p)  ;

\path[]%blockers
(p) edge (b1)
(p) edge (b2)         
(p) edge (bj)
(p) edge (bj+1)
;

\path[] %\Delta_{pp'bj+2c}
(bj+2c) edge  (p)
(p) edge (q')
(q') edge (bj+2c)
;

\path[dashed] %blocked
(q') edge[firstColor] (p')
(q') edge[firstColor] (b1)
(q') edge[firstColor] (b2)
(bj) edge[firstColor] (q')
(bj+1) edge[firstColor] (q')
;
\end{tikzpicture}
\caption{For $\{ p,p'\}\in E$, since $\Delta_{pp'b_{j+2c}}$ is oriented consistently, $|C_G(p,p')|\leq|\Delta_{pp'b_{j+2c}}|$.}  
\label{fig:2D-greedy-orientation-Cg}
\end{figure}

The following inequalities hold for these blockers: 
\begin{align}
& |p-b_i|\leq |p'-b_{i-1}|\label{eq:greedy-choice-i-i+1},\\
&  |p-b_i|\leq |p-b_k| \text{ for $1\leq i<k\leq j+2c$,}\label{eq:greedy-choice-order} \text{ and}\\
&  |p-b_{i-1}|\leq  |b_{i-2}-b_{i}| \text{ for $3\leq i\leq  j+2c$.}\label{eq:greedy-choice-i-i+2}  
\end{align}
Equation~\ref{eq:greedy-choice-order} is true, as otherwise $\{p,b_k\}$ would block $\{p',b_{i-1}\}$ instead of $\{p,b_i\}$.     

Equation~\ref{eq:greedy-choice-i-i+2} is explained as follows: due to convexity and planarity, $\{p,b_{i+1}\}\in E$ implies $\{b_i,b_{i+2}\}\notin E$. This means the greedy algorithm added $\{p,b_{i+1}\}$ and discarded $\{b_i,b_{i+2}\}$ in later iteration.

Let $\gamma>\pi$ be an arbitrary constant. We call an edge \emph{long} if its length is larger than $|p-p'|\cdot \gamma$. 
Because of Equation~\ref{eq:greedy-choice-order}, once an edge $\{p,b_i\}$ is long, every edge $\{p,b_k\}$ is also long, for $1\leq i<k\leq j+2c$. 
Let $b_{j}$ be the point such that $\{p,b_j\}$ is the \emph{last short} blocker, i.e.\
\begin{align}
& |p-b_{j}|\leq  |p-p'|\cdot \gamma \label{eq:greedy-last-shortblocker} \text{ and}\\
&  |p-b_k|> |p-p'|\cdot \gamma \text{ for all $j+1\leq  k\leq j+2c$.} \label{eq:greedy-longblockers}
\end{align} 
%  It is possible that already the first blocker $\{p,b_1\}$ is long.

By this, the length of the last blocker $\{p,b_{j+2c}\}$ depends on $\gamma$ and $c$:
\begin{align*}          
\label{eq:greedy-Radius}
|p-b_{j+2c}| & \overset{\text{eq.~\ref{eq:greedy-choice-i-i+1}}}{\leq} |p'-b_{j+2c-1}|\overset{\substack{\text{triangle}\\\text{inequality}}}{\leq} |p-p'|+ |p-b_{j+2c-1}|\\
&\overset{\text{eq.~\ref{eq:greedy-choice-i-i+1}}}{\leq} |p-p'|+ |p'-b_{j+2c-2}| \overset{\substack{\text{triangle}\\\text{inequality}}}{\leq}\hspace{-4mm} \ldots  \leq 2c\cdot|p-p'|+ |p-b_{j}| \\
& \overset{\text{eq.~\ref{eq:greedy-last-shortblocker}}}{\leq} (2c + \gamma ) \cdot |p-p'|
\tag{\theequation}\stepcounter{equation}
\end{align*}

Due to convexity, the points $b_{j+1},\ldots,b_{j+2c}$ must be contained in a half circle with centre $p$ and radius $|p-b_{j+2c}|$ (see  \Cref{fig:2D-greedy-circumference-blocker-i+2}). 
%annulus with centre $p$, the outer radius $|p-b_{j+2c}|$ and the inner radius $|p-b_{j+1}|$ (see  \Cref{fig:2D-greedy-circumference-blocker-i+2}). 
Therefore, their pairwise distances are bounded by the circumference of the half cycle. 
%Further, $b_j,\ldots,b_{j+2c}$ must be outside the cycle with center $p$ and radius $r= |p-b_j| \overset{\text{eq.~\ref{eq:greedy-longblockers}}}{>} |p-q'|\cdot \gamma$.
%Now, we combine these observations to prove $c,\gamma\in \bigo(1)$.
Upper and lower bounding the sum of the pairwise distances of tuples $b_{j+i}, b_{j+i+2}$ for $1\leq i\leq 2c-2$, it follows $c$ is bounded by a function dependent on $\gamma$:
\begin{align*}
&(2c-2) \cdot \gamma  \cdot |p-p'| = \sum_{i=2}^{2c-1} \gamma  \cdot |p-p'|  \\
&\overset{\text{eq.~\ref{eq:greedy-longblockers}}}{<}   \sum_{i=2}^{2c-1} |p-b_{j+i}| \overset{\text{eq.~\ref{eq:greedy-choice-i-i+2}}}{\leq}  \sum_{i=1}^{2c-2} |b_{j+i}-b_{j+i+2}|\\
&\leq \pi \cdot |p-b_{j+2c}| \overset{\text{eq.~\ref{eq:greedy-Radius}}}{\leq}  \pi \cdot (2c+\gamma) \cdot |p-p'|  \\    
% \iff~ &(2c-2)  \cdot \gamma < \pi (2c+\gamma) \\ 
%  \iff~&2c\gamma - 2 \gamma < 2c\pi + \pi \gamma\\ 
%   \iff~&2c\gamma - 2c\pi< \pi \gamma  + 2 \gamma\\ 
%    \iff~&c(\gamma - \pi)< \frac{\gamma (\pi + 2)}{2}\\ 
&\iff~ c  <  \gamma \cdot \frac{\pi  + 2}{2\gamma - 2\pi} \quad\text{for $\gamma>\pi$.}
\end{align*}
For $\gamma=\bigo(1)$, it holds that $c =\bigo(1)$. Thus, the length of the last blocker $\{p,b_{j+2c}\}$ is bounded by a constant times $|p-p'|$.       
\begin{figure}[ht]
\centering
\begin{tikzpicture}[every edge/.style={graph edge, semithick}]
\node [dot,label={-90:{$p$}}] (p) at (0,-2) {};
\node [dot,label={-90:{$b_{j+1}$}}] (bj) at (-1.8,-0.6) {};
\node [dot,label={90:{$b_{j+2}$}}] (bj+1) at (-1.4,0.4) {};
\node [dot,label={90:{$b_{j+3}$}}] (bj+2) at (0,1) {};
\node [dot,label={0:{$b_{j+2c-2}$}}] (bj+2c-2) at (2.2,0.4) {};
\node [dot,label={0:{$b_{j+2c-1}$}}] (bj+2c-1) at (3.2,-1.2) {};
\node [dot,label={0:{$b_{j+2c}$}}] (bj+2c) at (3.7,-3.2) {};
\draw  (p) edge (bj);
\draw  (p) edge (bj+2);
\draw  (p) edge (bj+2c-2);
\draw  (p) edge (bj+2c);
\draw [dashed] (bj) edge[firstColor] (bj+2);
\draw [loosely dotted, thick] (bj+2) edge (bj+2c-2);
\draw [dashed] (bj+2c-2) edge[firstColor] (bj+2c);        

\draw  (p) edge (bj+1);
\draw  (p) edge (bj+2c-1);

%  \tkzDrawCircle[dashed,opacity=.5](p,bj+2c)
% \tkzDrawCircle[dashed,opacity=.5](p,bj)
\tkzDrawArc[dashed,opacity=.5](p,bj+2c)(bj)
\end{tikzpicture}
%\caption{The points $b_{j+1},\ldots,b_{j+2c}$ lie in the annulus with centre $p$, the outer radius $|p-b_{j+2c}|$ and the inner radius $|p-b_{j+1}|$.}
\caption{The points $b_{j+1},\ldots,b_{j+2c}$ lie in a half circle with centre $p$ and radius $|p-b_{j+2c}|$.}
\label{fig:2D-greedy-circumference-blocker-i+2}
\end{figure}

For i) $\{p,p'\}\in {E}$, we conclude that, since $\Delta_{pp'b_{j+2c}}$ is a consistently oriented triangle in~$G$,  the length of the shortest closed walk containing $p$ and $p'$, is bounded by 
\begin{align*}
\label{eq:bound-C_g-edge-inside}
|C_G(p,p')| &\leq |\Delta_{pp'b_{j+2c}}| = |p-p'|+|p'-b_{j+2c}|+|p-b_{j+2c}|\\
&\overset{\text{triangle inequality}}{\leq}  2|p-p'|+2 |p-b_{j+2c}|\\
&\overset{\text{eq.~\ref{eq:greedy-Radius}}}{\leq}   2|p-p'|+2 |p-p'|\cdot (2c+\gamma) = |p-p'| \cdot (4c+2 \gamma +2).
\end{align*}
This permits a constant dilation for the case $\{p,p'\}\in E$ as       
\[\odil(p,p')=\frac{|C_G(p,p')|}{|\Delta(p,p')|} \leq \frac{|p-p'| \cdot (4c+2 \gamma +2)}{|\Delta_{pp'q}|}\leq  4c+2 \gamma +2= \bigo(1).\]

% \item \label{case:alg-2D-greedy-p-p1-not-inside}
For ii) $\{p,p'\}\notin {E}$, due to the $\alpha$-diamond property, there is an undirected path $\Pi$ from $p$ to $p'$ in $T$ of length ${|\Pi|}\leq t_g \cdot |p-p'|$.
For each path edge $\{q',q''\}\in \Pi$ it holds that $\{q',q''\}\in E$. So, we apply case~i) for each path edge. % which gives factor $\lambda= 4c+2 \gamma +2 = \bigo(1)$ for each. 
The dilation for the case $\{p,p'\}\notin E$ is bounded by
\begin{align*}
\odil(p,p')=&\frac{|C_G(p,p')|}{|\Delta(p,p')|}\leq\frac{ \sum_{ \{q',q''\}\in \Pi} |C_G(q',q'')|}{|\Delta(p,p')|}& \\
\overset{\text{case i)}}{\leq}& \frac{ \sum_{ \{q',q''\}\in \Pi}  |q'-q''|  \cdot (4c+2 \gamma +2) }{|\Delta(p,p')|} = \frac{ (4c+2 \gamma +2) \cdot |\Pi|}{|\Delta(p,p')|}\\ 
\leq& \frac{ (4c+2 \gamma +2)  \cdot t_g \cdot |p-p'|}{|\Delta(p,p')|} \leq (4c+2 \gamma +2) \cdot t_g = \bigo(1). \tag*{\qedhere} %\qedhere
\end{align*} 
% \begin{align*}
%     |C_G(p,p')|&\leq \sum_{ \{q,q'\}\in \Pi} |C_G(q,q')|\\
%     &\overset{\text{case i)}}{\leq} \sum_{ \{q,q'\}\in \Pi} |q-q'|\cdot \lambda = \lambda \cdot \sum_{ q_i,q_{i+1}\in \Pi} |q_i-q_{i+1}| =  \lambda \cdot |\Pi|\\
%     &\leq \lambda \cdot t_g \cdot |p-p'|.
% \end{align*}

%  Merging all cases, the dilation is bounded by 
% \[
% \odil(p,p')\leq \max\{\underset{\text{i)} \Delta_{pq'p'}\in E }{2\cdot t_g},\underset{\text{i)} \Delta_{pp'b_{j+2c}}\in E}{4c+2 \gamma +2 },\underset{\text{ii)} \{p,p'\}\notin E}{(4c+2 \gamma +2) \cdot t_g }\}= (4c+2 \gamma +2) \cdot t_g = \bigo(1). \qedhere
% \]
\end{proof}

The aim of this proof was to show that the consistent orientation of the greedy triangulation provides a plane constant dilation spanner at all.    Since $c$ is a function depending on $\gamma> \pi$, the proven constant factor is minimised as
$$\underset{\gamma>\pi}{\min}~ \left( 4 \cdot \gamma \cdot \frac{\pi  + 2}{2\gamma - 2\pi} + 2\gamma+2 \right) \cdot t_g = (\pi+\sqrt{\pi(2+\pi)})\cdot t_g < 7.2 \cdot t_g,$$
where $t_g$ is the dilation of the greedy triangulation.
%% Calculation:  https://www.wolframalpha.com/input?i=minimize%5B%284x%28%5Cpi++-+2%29%2F%282x-+2%5Cpi%29+%2B+2x%2B2%29+%26%26++x%3E%5Cpi%5D&assumption=%7B%22C%22%2C+%22minimize%22%7D+-%3E+%7B%22Calculator%22%2C+%22dflt%22%7D&assumption=%7B%22F%22%2C+%22GlobalMinimizeCalculator%22%2C+%22curvefunction%22%7D+-%3E%224x%28%5Cpi++-+2%29%2F%282x-+2%5Cpi%29+%2B+2x%2B2%22&assumption=%7B%22F%22%2C+%22GlobalMinimizeCalculator%22%2C+%22domain%22%7D+-%3E%22%5Cpi%3Cx%22

\begin{corollary}
Let $t_g$ be the smallest upper bound on the dilation of a greedy triangulation.
By orienting the greedy triangulation of a convex two-dimensional point set $P$ consistently, we get a plane oriented $(7.2 \cdot t_g)$-spanner for $P$.
\end{corollary}
% \begin{proof}
%  \begin{align*}
%      \odil(p,p')&\leq \max\{\underset{\text{i)} \Delta_{pq'p'}\in E }{2\cdot t_g},\underset{\text{i)} \Delta_{pp'b_{j+2c}}\in E}{4c+2 \gamma +2 },\underset{\text{ii)} \{p,p'\}\notin E}{(4c+2 \gamma +2) \cdot t_g }\}= (4c+2 \gamma +2) \cdot t_g \\
%      & \leq  t_g  \cdot \underset{\gamma>\pi}{\min}  4 \pi  \frac{\gamma}{2 \gamma-  2\pi} + 2\gamma+2  \leq 3\cdot t_g,
%  \end{align*}
% where $t_g$ is the smallest upper bound on the dilation of a greedy triangulation. (Up to our knowledge, the best known upper bound is given in \cite{BoseLS07}.)
% \end{proof}

% \begin{align*}
%             t_g&\leq \frac{8(\pi - \alpha)^2}{\alpha ^2 \sin^2(\sfrac{\alpha}{4})} \text{ with $\alpha=\sfrac{\pi}{6}$} 
%             % &\leq \frac{8(\pi - \frac{\pi}{6})^2}{\left(\frac{\pi}{6}\right)^2 \sin^2(\frac{\frac{\pi}{6}}{4})}\\
%             % &\leq \frac{8(\frac{5\pi}{6})^2}{\frac{\pi^2}{36} \sin^2(\frac{\pi}{24})} \\
%             %  &\leq \frac{8\cdot \frac{25\pi^2}{36}}{\frac{\pi^2}{36} \sin^2(\frac{\pi}{24})} \\
%             %   &\leq \frac{\pi^2}{9} \cdot \frac{50}{\frac{\pi^2}{36} \sin^2(\frac{\pi}{24})} \\
%             %   &\leq \frac{\pi^2}{9} \cdot \frac{50\cdot 36}{\pi^2 \sin^2(\frac{\pi}{24})} \\
%                ~~\leq\frac{200}{\sin^2(\frac{\pi}{24})} \text{~\cite{BoseLS07}.}%\approx 11739.1
%         \end{align*}
%%%%%%%%%%%%%%%%%

\subsection{Other Triangulations}\label{sec:other_triangulations}

As the greedy triangulation leads to a plane oriented $\bigo(1)$-spanner for convex point sets, the question arises whether the $\alpha$-diamond property implies the existence of oriented spanners. 

We already have seen that this is not the case for the greedy triangulation on general point sets. Likewise, the minimum weight triangulation has the $\alpha$-diamond property~\cite{KeilGutwin92} but does not yield an $\bigo(1)$-spanner in general. This can be seen by the same example as for the greedy triangulation in \Cref{fig:greedy-non-convex-omega-n}. By essentially the same argument, the minimum weight triangulation will have the same edges incident to $p_1$, which results in an oriented dilation of $\Omega(n)$. Whether the minimum weight triangulation is an oriented $\bigo(1)$-spanner for convex point sets, we leave as an open problem.

For the Delaunay triangulation the situation is even worse. The Delaunay triangulation has the $\alpha$-diamond property~\cite{DrysdaleMS01} and is the basis for many undirected plane spanner constructions~\cite{Bose.2013}. However, for $n \geq 4$ its oriented dilation can be arbitrary large, no matter how we orient its edges, even for convex point sets. 

\begin{figure}%[8]{r}{0.4\linewidth}\vspace{-8mm}
\centering%
\begin{tikzpicture}[every edge/.style={graph edge, semithick}, scale=.8]
\node [dot,label={-90:{$p_3$}}]  (3) at (0,0) {};
\node [dot,label={-5:{$p_2$}}] (2) at ($(3)+(.7,0.1)$) {};
\node [dot,label={-175:{$p_4$}}] (4) at ($(3)+(-.7,.1)$) {};

\tkzCircumCenter(2,3,4)%
\tkzGetPoint{O}%
\tkzDrawCircle[dashed,opacity=.5](O,2)%

\node [dot,label={-5:{$p_1$}}] (1) at ($(3)+(0,4.8)$) {};

\path
(1) edge (2)
(1) edge (3)
(1) edge (4)
(2) edge (3)
(3) edge (4)
;
%   \draw [opacity=.5]  ($(2)!.2!(3)$) edge[bend right=50] node [above right] {$\beta$} ($(4)!.2!(3)$);
\end{tikzpicture}%
\caption{Delaunay triangulation for a convex point set. It is also a minimum dilation triangulation.}%
\label{fig:delaunay-arbitrary-worse}
\end{figure}

\Cref{fig:delaunay-arbitrary-worse} shows the Delaunay triangulation $T=(P,E)$ of a convex point set of size~$4$. The shortest closed walk $C_G(p_2,p_4)$ contains $p_1$, for any orientation of $T$. We can place $p_2, p_3, p_4$ arbitrarily close to each other without decreasing the radius of the circle through them. By placing $p_1$ in the circle and sufficiently far away from $p_2, p_3$ and $p_4$, the oriented dilation
$\odil(p_2,p_4)\geq\frac{|C_G(p_2,p_4)|}{|\Delta_{p_2p_3p_4}|}$ can be made arbitrary large. The example can be modified to include more points without changing $\odil(p_2,p_4)$.

Considering Delaunay and minimum weight triangulation, we conclude that the $\alpha$-diamond property is not sufficient to find  oriented $\bigo(1)$-spanners.

Finally, we consider orienting the undirected minimum dilation triangulation. However, essentially the same example as for the Delaunay triangulation shows that there are convex point sets for which any orientation has arbitrarily large oriented dilation. In \Cref{fig:delaunay-arbitrary-worse}, consider the case in which we have fixed all positions except for the $y$-coordinates of $p_2$ and $p_4$. Let $y_0$ be the $y$-coordinate of $p_3$, and let $y_0 + \varepsilon$ be the $y$-coordinate of $p_2$ and $p_4$. By decreasing $\varepsilon>0$, we can make sure that the minimum dilation triangulation chooses the same diagonal as the Delaunay triangulation, since the undirected dilation between $p_2$ and $p_4$ by the path through $p_3$ becomes arbitrarily close to~$1$. However, as in the Delaunay triangulation, the oriented dilation between $p_2$ and $p_4$ can be made arbitrarily large. 

Thus, the quest for a triangulation of small oriented dilation for non-convex point sets remains open. 

%As already mentioned, computing the minimum plane undirected spanner lead to the minimum dilation triangulation. 
%So, we consider again the point set $P=\{p_1,p_2,p_3,p_4\}$, which is shown in \Cref{fig:delaunay-arbitrary-worse}. By arbitrarily decreasing the $y$-coordinates of $p_2$, $p_3$ and $p_4$, such that these are almost co-linear, the minimum dilation triangulation $T=(P,E)$ has the same edge set like the Delaunay triangulation.
%Analogous, the shortest closed walk $C_G(p_2,p_4)$, which contains $p_1$ for any orientation of $T$, can be arbitrary large by placing $p_1$ arbitrary far away. 
%Therefore, orienting the minimum dilation triangulation will not give a constant dilation spanner even for convex point sets.

\section{Conclusion and Outlook}

Motivated by applications of geometric spanners, we introduced the concept of oriented geometric spanners.
We provided a wide range of extremal and algorithmic results for oriented spanners in one and two dimensions. 

%Remarkably, orienting the Delaunay triangulation and the minimum weight triangulation does not lead to oriented constant dilation spanners, even for points in convex position. 
%Contrary, orienting the greedy triangulation yields an $\bigo(1)$-spanner for point sets in convex position.

Intriguingly, orienting the greedy triangulation yields a plane $\bigo(1)$-spanner for point sets in convex position, but not for general point sets. Furthermore, other natural triangulations like the Delaunay and the minimum weight triangulation do not lead to plane constant dilation spanners.

This raises the question of whether a plane constant dilation spanner exists and can be computed efficiently. If this is not possible, what is the lowest dilation that we can guarantee? Until now, even showing whether the greedy triangulation yields an $\bigo(n)$-spanner remains open.

As the concept of oriented spanners is newly introduced, it opens up many new avenues of research, for example:
\begin{itemize}
\item We know that the minimum one-page plane oriented spanner achieves a dilation of at most~$5$ (since this is the bound for the greedy algorithms) and that there are point sets where it has a dilation of~$2$. What is its worst-case dilation $2\leq t \leq 5$? Can we compute it faster?
\item We constructed a two-page plane oriented $2$-spanner for any one-dimensional point set. Is $2$ a tight upper bound on the dilation of a minimum two-page plane spanner? Is there an efficient algorithm to compute such a spanner?         
%  \todo{We construct a two-page plan oriented 2-spanner for any one-dimensional point set. Is there an efficient algorithm to compute a two-page oriented $t$-spanner for any $t < 2$? }
\item For two-dimensional point sets, the question of bounding minimum oriented dilation already arises without restricting to plane graphs. What is the worst-case dilation $2\sqrt{3}-2\leq t \leq 2$ of the minimum dilation oriented complete graph?
\item Given an undirected geometric graph, can we efficiently compute an orientation minimising the dilation? For which graph classes is this possible?
\item  While undirected dilation compares the shortest path from $p$ to $p'$ in $G$ with the edge between them in the complete graph $K_n$, oriented dilation compares the shortest closed walk through $p$ to $p'$ in $G$ to the shortest triangle in $K_n$. An analogous measure in undirected graphs, \emph{cyclic dilation}, would compare the shortest cycle in $G$ to the shortest triangle in $K_n$. Can we efficiently compute sparse graphs, in particular triangulations, of low cyclic dilation? Another measure of interest would be \emph{detour dilation}, which compares the shortest path not using the edge $\{p,p'\}$ in $G$ (i.e.\ the shortest path in $G - \{p,p'\}$) with the triangle in $K_n$. Low detour dilation is a necessary condition for low oriented dilation, and actually at the core of our analysis of the greedy triangulation. Is there a triangulation with constant detour dilation? Can it be oriented to obtain constant oriented dilation?
\item In the applications mentioned, it is desirable to reduce the number of bi-directional edges but it may not be necessary to avoid them completely. This opens up a whole new set of questions on the trade-off between directed dilation and the number of bi-directional edges. %For instance, given a bound $t$ on the directed dilation, compute the $t$-spanner with as few bi-directional edges as possible (and no bound on the number of oriented edges). 
For instance, given a parameter $t$, compute the directed $t$-spanner with as few bi-directed edges as possible (and no bound on the number of oriented edges). 
Or, given a plane graph with certain edges marked as \emph{one-way}, compute the orientation of these edges that minimises the directed dilation (while all other edges are bi-directional). For which families of graphs can this be done efficiently? For which is this problem NP-hard?
\end{itemize}

\bibliography{bib-arxiv} %% Please use bibtex, 

@PREAMBLE{"\newcommand{\SortNoop}[1]{}"}

@String { algorithmica = {Algorithmica} }

@String { and          = { and } }

@String { camb-u-p     = {Cambridge University Press} }

@String { cccg         = { Canad. Conf. Comput. Geom. (CCCG)} }

@String { cgta         = {Comput. Geom. Theory Appl.} }

@String { dam           = {Discret. Appl. Math.} }

@String { dcg          = {Discrete Comput. Geom.} }

@String { elsevier     = {Elsevier} }

@String { esa          = { Annu. European Sympos. Algorithms (ESA)} }

@String { ijcga        = {Internat. J. Comput. Geom. Appl.} }

@String { jan          = {January} }

@String { josis         = {J. Spatial Information Science} }

@String { jun          = {June} }

@String { lncs         = {Lecture Notes in Computer Science} }

@String { proc         = {Proc. } }

@String { siam         = {SIAM} }

@String { sicomp       = {SIAM J. Comput.} }

@String { simon        = {Simon \& Schuster} }

@String { soda         = { Annu. ACM-SIAM Sympos. Discrete Algorithms (SODA)} }

@String { springer     = {Springer-Verlag} }

@string{wads = " Internat. Sympos. Algorithms and Data Struct. "}

@inproceedings{brandt-mdt-2014,
	author       = {Al{\'{e}}x F. Brandt and
	Miguel M. Gaiowski and
	Cid C. de Souza and
	Pedro J. de Rezende},
	title        = {Minimum Dilation Triangulation: Reaching Optimality Efficiently},
	booktitle    = proc # "26th" # cccg,
	year         = {2014},
}

@inproceedings{burkhart2004does,
	title={Does topology control reduce interference?},
	author={Burkhart, Martin and Von Rickenbach, Pascal and Wattenhofer, Rogert and Zollinger, Aaron},
	booktitle={Proc.\ 5th ACM Internat. Sympos. Mobile Ad Hoc Networking and Computing},
	pages={9--19},
	year={2004},
	doi={10.1145/989459.989462   }
}

@article{schindelhauer2007geometric,
	title={Geometric spanners with applications in wireless networks},
	author={Schindelhauer, Christian and Volbert, Klaus and Ziegler, Martin},
	journal= cgta,
	volume={36},
	number={3},
	pages={197--214},
	year={2007},
	doi ={10.1016/j.  comgeo.2006.02.001 },
	publisher=elsevier
}

@inproceedings{BoseLS07,
	author       = {Bose, Prosenjit and
	Lee, Aaron and
	Smid, Michiel},
	title        = {On Generalized Diamond Spanners},
	editor= {Dehne, Frank
	and Sack, J{\"o}rg-R{\"u}diger
	and Zeh, Norbert},
	booktitle    = proc # "10th" # wads,
	series       = lncs,
	volume       = {4619},
	pages        = {325--336},
	publisher    = springer,
	year         = {2007},
	doi          = {10.1007/978-3-540-73951-7_29},
}

@article{aronov2011connect,
	title={Connect the dot: Computing feed-links for network extension},
	author={Aronov, Boris and Buchin, Kevin and Buchin, Maike and Jansen, Bart and De Jong, Tom and Kreveld, Marc van and L\"offler, Maarten and Luo, Jun and Speckmann, Bettina and Silveira, Rodrigo I.},
	journal=josis,
	volume={3},
	pages={3--31},
	year={2011},
	doi={10.5311/JOSIS.2011.3.47 }
}

@article{dobson2014sparse,
	title = {Sparse roadmap spanners for asymptotically near-optimal motion planning},
	author = {Dobson, Andrew and Bekris, Kostas},
	journal = {Internat.\ J.\ Robotics Research},
	number = {1},
	pages = {18--47},
	volume = {33},
	year = {2014},
	doi = {10.1177/0278364913498 }
}

@book{Narasimhan.2007,
	author    = {Giri Narasimhan and
	Michiel Smid},
	title     = {Geometric Spanner Networks},
	publisher = camb-u-p,
	year      = {2007},
	doi = {10.1017/CBO9780511546884  },
}

@article{Bose.2013,
	author = {Bose, Prosenjit and Smid, Michiel},
	year = {2013},
	title = {On Plane Geometric Spanners: A Survey and Open Problems},
	pages = {818--830},
	volume = {46},
	number = {7},
	journal = cgta,
	doi = {10.1016/j.   comgeo.2013.04.002 }
}

@article{Giannopoulos.2010,
	author    = {Panos Giannopoulos and
	Rolf Klein and
	Christian Knauer and
	Martin Kutz and
	D{\'{a}}niel Marx},
	title     = {Computing Geometric Minimum-Dilation Graphs is {NP}-Hard},
	journal   = ijcga,
	volume    = {20},
	number    = {2},
	pages     = {147--173},
	year      = {2010},
	doi = {10.1142/S0218195910003244      }
}

@inproceedings{Cowen.1999,
	author = {Cowen, Lenore J. and Wagner, Christopher G.},
	title = {Compact Roundtrip Routing for Digraphs},
	editor       = {Robert Endre Tarjan and
	Tandy J. Warnow},
	pages = {885--886},
	publisher = {{SIAM}},
	booktitle = proc # "10th" # soda, 
	year = {1999},
	doi = {10.5555/314500.315068}, 
}

@article{Cowen.2004,
	author = {Cowen, Lenore J. and Wagner, Christopher G.},
	year = {2004},
	title = {Compact roundtrip routing in directed networks},
	pages = {79--95},
	volume = {50},
	number = {1},
	journal = {J.\ Algorithms},
	doi = {10.1016/j          .jalgor.2003.08.001},
}

@article{KeilGutwin92,
	author       = {J. Mark Keil and
	Carl A. Gutwin},
	title        = {Classes of Graphs Which Approximate the Complete Euclidean Graph},
	journal      = dcg,
	volume       = {7},
	pages        = {13--28},
	year         = {1992},
	doi = {10.1007/BF02187821     }
}

@inproceedings{DasJ89,
	author       = {Gautam Das and
	Deborah Joseph},
	title        = {Which Triangulations Approximate the Complete Graph?},
	editor={Djidjev, Hristo},
	doi ={10.1007/3-540-51859-2_15     },
	booktitle    = {Optimal Algorithms},
	series       = lncs,
	volume       = {401},
	pages        = {168--192},
	publisher    = springer,
	year         = {1989}
}

@article{DrysdaleMS01,
	author       = {Drysdale, Robert L. (Scot)  and
	McElfresh,Scott A.  and
	Snoeyink, Jack },
	title        = {On exclusion regions for optimal triangulations},
	journal      = dam,
	volume       = {109},
	number       = {1-2},
	pages        = {49--65},
	year         = {2001},
	doi =  {10.1016/S0166-218X(00)00236-5      }
}

@incollection{Eppstein00,
	author       = {David Eppstein},
	title        = {Spanning Trees and Spanners},
	booktitle    = {Handbook of Computational Geometry},
	pages        = {425--461},
	publisher = {North-Holland},
	year         = {2000},
	doi= {h10.1016/B978-044482537-7/50010-3},
}

@article{LevcopoulosL92,
	author       = {Christos Levcopoulos and
	Andrzej Lingas},
	title        = {Fast Algorithms for Greedy Triangulation},
	journal      = {{BIT}},
	volume       = {32},
	number       = {2},
	pages        = {280--296},
	year         = {1992},
	doi = {10.1007/BF01994882      },
}

@article{ChungLR87.bookembeddings, 
	author = {Chung, Fan R. K. and Leighton, Frank Thomson and Rosenberg, Arnold L.},
	title = {Embedding Graphs in Books: A Layout Problem with Applications to VLSI Design},
	year = {1987},
	issue_date = {Jan 1987},
	publisher = {Society for Industrial and Applied Mathematics},
	volume = {8},
	number = {1},
	doi = {10.1137/0608002                },
	journal = {SIAM J. Algebraic Discrete Methods},
	pages = {33–58},
}

@article{Yannakakis89.bookembeddings,
	author       = {Mihalis Yannakakis},
	title        = {Embedding Planar Graphs in Four Pages},
	journal      = {J. Comput. Syst. Sci.},
	volume       = {38},
	number       = {1},
	pages        = {36--67},
	year         = {1989},
	doi          = {10.1016/0022-0000(89)90032-9              },
}

@article{LEVCOPOULOS1999197,
	title = {The greedy triangulation can be computed from the {D}elaunay triangulation in linear time},
	journal = cgta,
	volume = {14},
	number = {4},
	pages = {197-220},
	year = {1999},
	issn = {0925-7721},
	doi = {10.1016/S0925-7721(99)00037-1       },
	author = {Christos Levcopoulos and Drago Krznaric},
}

@article{BekosGR16.bookembeddings,
	author       = {Michael A. Bekos and
	Martin Gronemann and
	Chrysanthi N. Raftopoulou},
	title        = {Two-Page Book Embeddings of 4-Planar Graphs},
	journal      = algorithmica,
	volume       = {75},
	number       = {1},
	pages        = {158--185},
	year         = {2016},
	doi          = {10.1007/s00453-015-0016-8            },
}

@article{DujmovicW04.bookembeddings,
	author       = {Vida Dujmovic and
	David R. Wood},
	title        = {On Linear Layouts of Graphs},
	journal      = {Discret. Math. Theor. Comput. Sci.},
	volume       = {6},
	pages        = {339--358},
	year         = {2004},
	doi          = {10.46298/dmtcs.317},
}

@article{AkitayaBB22,
	author       = {Hugo A. Akitaya and
	Ahmad Biniaz and
	Prosenjit Bose},
	title        = {On the spanning and routing ratios of the directed {\(\Theta\)}\({}_{\mbox{6}}\)-graph},
	journal      = {Comput. Geom.},
	volume       = {105-106},
	pages        = {101881},
	year         = {2022},
	doi          = {10.1016/j. comgeo.2022.101881},
}

@inproceedings{BuchinGKPRRW23,
	author       = {Kevin Buchin and
	Joachim Gudmundsson and
	Antonia Kalb and
	Aleksandr Popov and
	Carolin Rehs and
	Andr{\'{e}} van Renssen and
	Sampson Wong},
	editor       = {Inge Li G{\o}rtz and
	Martin Farach{-}Colton and
	Simon J. Puglisi and
	Grzegorz Herman},
	title        = {Oriented Spanners},
	booktitle    = {31st Annual European Symposium on Algorithms, {ESA} 2023, September
	4-6, 2023, Amsterdam, The Netherlands},
	series       = {LIPIcs},
	volume       = {274},
	pages        = {26:1--26:16},
	publisher    = {Schloss Dagstuhl - Leibniz-Zentrum f{\"{u}}r Informatik},
	year         = {2023},
	doi          = {10.4230/LIPIcs.ESA.2023.26},
}

@article{ItaiPS82.hamilton,
	author       = {Alon Itai and
	Christos H. Papadimitriou and
	Jayme Luiz Szwarcfiter},
	title        = {Hamilton Paths in Grid Graphs},
	journal      = sicomp,
	volume       = {11},
	number       = {4},
	pages        = {676--686},
	year         = {1982}
}

%%%%%%%%%%%%%%%%%%%%%%%%%%%%%%%

%%%%%%%%%%%%%%%%%%%%%

%\pagebreak
%%%%%%%%%%%%%%%%%%%%%%%%%%%%%%

% \section{Correctness of orienting greedy triangulation consistently}
% \twoDgreedyConvex*

\end{document}